\documentclass[11pt]{article}

\usepackage[margin=1in]{geometry}
\usepackage[utf8]{inputenc}
\usepackage[T1]{fontenc}

\usepackage{amsmath,amssymb,amsthm}
\usepackage{booktabs}
\usepackage[ruled]{algorithm2e}

\SetAlFnt{\small}
\SetAlCapFnt{\small}
\SetAlCapNameFnt{\small}
\SetAlCapHSkip{0pt}
\IncMargin{-\parindent}

\usepackage{thm-restate}
\usepackage{enumitem}
\usepackage{mathtools}
\usepackage{multirow}
\usepackage{dsfont}
\usepackage{upgreek}
\usepackage{comment}
\usepackage{bbm}
\usepackage{makecell}
\usepackage{hyperref}
\usepackage{xcolor}
\usepackage[numbers,sort&compress]{natbib}
\setcitestyle{numbers,square}
\allowdisplaybreaks[1]

\hypersetup{
    colorlinks=true,
    linkcolor=blue,
    citecolor=blue,
    urlcolor=blue,
}

\newcommand{\optimizer}{\mathcal{E}}

\newcommand{\probability}{\mathbb{P}}

\newcommand{\D}{\mathcal{D}}

\newcommand{\U}{\mathcal{U}}
\newcommand{\E}{\mathbb{E}}
\newcommand{\s}{\sigma_{\beta}}
\newcommand{\C}{\textsc{Copeland} }
\newcommand{\Cc}{\textsc{Copeland}}
\newcommand{\B}{\textsc{Borda} }
\newcommand{\Bb}{\textsc{Borda}}
\newcommand{\R}{\textsc{RandomDictator} }
\newcommand{\Rr}{\textsc{RandomDictator}}
\newcommand{\Plu}{\textsc{Plurality} }
\newcommand{\Plup}{\textsc{Plurality}}
\newcommand{\PluV}{\textsc{PluralityVeto} }
\newcommand{\PluVv}{\textsc{PluralityVeto}}
\newcommand{\PPV}{\textsc{PrunedPluralityVeto} }
\newcommand{\PPVv}{\textsc{PrunedPluralityVeto}}
\newcommand{\ML}{\textsc{MaximalLotteries} }

\DeclareMathOperator*{\argmax}{arg\,max}

\theoremstyle{plain}
\newtheorem{theorem}{Theorem}[section]
\newtheorem{lemma}[theorem]{Lemma}

\newtheorem{corollary}[theorem]{Corollary}

\newtheorem{openq}{Open Question}

\theoremstyle{definition}
\newtheorem{definition}[theorem]{Definition}

\theoremstyle{remark}

\title{\textbf{Utilitarian Distortion Under Probabilistic Voting}}

\author{
    Hamidreza Alipour\thanks{Authors listed in alphabetical order.}\\
    \textit{Sharif University of Technology}\\
    \texttt{hmr.alipour@gmail.com}
    \and 
    Mohak Goyal\\
    \textit{Stanford University}\\
    \texttt{goyalmohak2@gmail.com}
}

\date{\today}

\begin{document}

\maketitle

    \begin{abstract}

The utilitarian distortion framework evaluates voting rules by their worst-case efficiency loss when voters have cardinal utilities but express only ordinal rankings. Under the classical model, a longstanding tension exists: Plurality, which discards most of the ranking information and suffers from the spoiler effect, achieves optimal $\Theta(m^2)$ distortion among deterministic rules, while normatively superior rules like Copeland and Borda have unbounded distortion. We resolve this tension under probabilistic voting with the Plackett-Luce model, where rankings are noisy reflections of utilities governed by an inverse temperature parameter $\beta$, and show that the distortion landscape changes dramatically. Copeland and Borda both achieve at most $\beta\frac{1+e^{-\beta}}{1-e^{-\beta}}$ distortion, independent of the number of candidates $m$, and within a factor of 2 of the lower bound for randomized rules satisfying the probabilistic Condorcet loser criterion known from prior work. This improves upon the prior $O(\beta^2)$ bound for Borda. These upper bounds are nearly tight: prior work establishes a $(1-o(1))\beta$ lower bound for Borda, and we prove a $(1-\epsilon)\beta$ lower bound for Copeland for any constant $\epsilon >0$. In contrast, rules that rely only on top-choice information fare worse: Plurality has distortion $\Omega(\min(e^\beta, m))$ and Random Dictator has distortion $\Theta(m)$. Additional `veto' information is also insufficient to remove the dependence on $m$; Plurality Veto and Pruned Plurality Veto have distortion $\Omega(\beta \ln m)$. We also prove a lower bound of $(\frac{5}{8}-\epsilon)\beta$ (for any constant $\epsilon >0$) for all deterministic finite-precision tournament-based rules, a class that includes Copeland and any rule based on pairwise comparison margins rounded to fixed precision. Our results show that the distortion framework aligns with normative intuitions once the probabilistic nature of real-world voting is taken into account.
\end{abstract}


\section{Introduction}

The \emph{utilitarian distortion} framework, introduced by \citet{procaccia2006distortion}, provides a way to evaluate voting rules when voters have cardinal utilities but can only express ordinal rankings. The \emph{distortion} of a voting rule is the worst-case ratio between the social welfare of the optimal candidate and that of the selected candidate. A rule with distortion $\alpha$ guarantees that the selected candidate achieves at least a $1/\alpha$ fraction of the optimal welfare, regardless of the underlying utility profile. This framework has attracted significant attention as it provides robust, distribution-free guarantees on the quality of collective decisions.

A notable tension exists between distortion and normative properties of voting rules. With $m$ candidates, under the classical model, where voters rank candidates deterministically according to their utilities, \Plu achieves $\Theta(m^2)$ distortion in the unit-sum model, which is optimal among deterministic rules~\citep{caragiannis2011voting}. In contrast, in the same unit-sum model, \C and \B have unbounded distortion despite their strong normative appeal~\citep{procaccia2006distortion, gkatzelis2023best}. That is, there exists no finite function of $m$ bounding the distortion over all numbers of voters. \C satisfies the Condorcet criterion, while \B has been axiomatically characterized as uniquely satisfying several desirable properties~\citep{young1974axiomatization, maskin2025borda}. Yet under the distortion criterion, \Plup, which suffers from the spoiler effect, outperforms these rules. While the distortion guarantee is appealing in principle, these counterintuitive findings have limited its influence on practical voting system design.

We show that this tension is resolved under a more realistic model of voter behavior. Following the recent work of \citet{goyal2025metric} on metric distortion and \citet{golz2025distortion} on utilitarian distortion, we study voting under the PL model~\citep{luce2005individual, plackett1975analysis}, where each voter's ranking is a probabilistic function of their utilities. Under this model, the probability that a voter ranks candidate $A$ above candidate $B$ is given by the logistic sigmoid of the difference in utilities. The parameter $\beta \geq 1$ captures the inverse \emph{temperature}: higher $\beta$ corresponds to less noise, so rankings more closely reflect true utilities. This model, widely used in econometrics and machine learning, captures the realistic assumption that voters may have uncertainty or make errors when translating preferences into rankings.

Under probabilistic voting, the distortion landscape changes dramatically: \C and \B achieve distortion independent of the number of candidates $m$ and linear in $\beta$, while \Plup, \PluVv, \PPVv, and \R (Definitions in \S~\ref{sec:voting_rules}) have distortion growing with $m$ or exponential in $\beta$. The distortion framework thus aligns with normative intuitions once we account for the probabilistic nature of real-world voting.

Our work contributes to the growing literature connecting social choice theory and AI alignment~\citep{conitzer2024social, mishra2023ai, sorensen2024roadmap}. In reinforcement learning from human feedback (RLHF)~\citep{ouyang2022training, stiennon2020learning}, human preferences are elicited through pairwise comparisons, typically modeled using Bradley-Terry, which is the pairwise marginal of the Plackett-Luce (PL) model. Methods based on aggregating pairwise comparisons, including Borda-style approaches, have been used in practice~\citep{siththaranjan2023distributional}. Our results provide theoretical support for such methods in the form of robust welfare guarantees.

\subsection{Our Contributions}
\label{sec:contributions}

Our main results are summarized in Table~\ref{tab:results}. We establish that probabilistic voting fundamentally changes which voting rules achieve good distortion. Formal definitions appear in \S\ref{sec:prelims}.

\textit{Upper Bounds for \C and \Bb.}
We prove that both rules achieve distortion at most $\beta \frac{1+e^{-\beta}}{1-e^{-\beta}}$, which, crucially, is independent of the number of candidates $m$ (Theorems~\ref{thm:copeland_ub} and~\ref{thm:borda_ub}). This places these rules within a factor of 2 of the lower bound $\frac{\beta}{2} \frac{1+e^{-\beta}}{1-e^{-\beta}}$ given by \citet{golz2025distortion} for \textit{randomized} rules satisfying the probabilistic Condorcet loser criterion (PCLC) (Definition~\ref{def:pclc}). Our bound for \B improves upon the $O(\beta^2)$ bound from~\citet{golz2025distortion}.

Perhaps surprisingly, \C and \B achieve identical distortion bounds. This contrasts with metric distortion, where \C significantly outperforms \B~\citep{anshelevich2015approximating, goyal2025metric}. This result is noteworthy since \B has been the method of choice in preference learning~\citep{siththaranjan2023distributional}, which is arguably more similar to the utilitarian setting than metric distortion.

\paragraph{Near-Tight Lower Bounds for \C and \Bb.}
We prove that \C has distortion at least $(1-\epsilon)\beta$ for any constant $\epsilon >0$, as $\beta \to \infty$ (Theorem~\ref{thm:copeland_lb}). Combined with the $(1-o(1))\beta$ lower bound for \B due to \citet{golz2025distortion}, this shows our upper bounds are near-tight.

\paragraph{Bounds for \Plup, \PluVv, \PPVv, and \Rr.}
We prove that rules that do not fully utilize information have high distortion:
\begin{itemize}[leftmargin=*]
    \item \Plu has distortion $\Omega(\min(e^\beta, m))$ (Theorem~\ref{thm:plurality_lb}). The proof leverages the fact that it continues to suffer from the spoiler effect under probabilistic voting.
    \item \PluV and \PPV have distortion $\Omega(\beta \ln m)$ (Theorem~\ref{thm:pluV_lb} and Corollary~\ref{corr:ppv}). This contrasts with classical frameworks, since \PluV achieves optimal metric distortion of 3 among deterministic rules \citep{kizilkaya2022plurality} and \PPV is orderwise optimal among deterministic rules for both utilitarian and metric distortion \citep{gkatzelis2023best}.
    \item \R has distortion $\Theta(m)$ (Theorem~\ref{thm:rd_lb}).
\end{itemize}
We also establish an upper bound of $O(\min(\frac{e^{2\beta}}{\beta}, \frac{me^\beta}{\ln m}))$ for \Plu (Theorem~\ref{thm:plurality_ub}).

\paragraph{Lower Bound for Finite-Precision Tournament Rules.}
We prove that any deterministic voting rule whose output depends only on a finite-precision weighted tournament graph has distortion at least $(\frac{5}{8} - \epsilon)\beta$ for sufficiently large $\beta$ and any constant $\epsilon >0$ (Theorem~\ref{thm:detTournamentGraph_lb}). This class includes rules based on the unweighted tournament graph (who beats whom) and rules that round margins to a fixed grid.
This establishes a gap between deterministic and randomized tournament-based rules since the randomized \ML rule has distortion of $\frac{\beta}{2} \frac{1+e^{-\beta}}{1-e^{-\beta}}$~\citep{golz2025distortion}.

\paragraph{Technical Overview.}
Our proofs require careful analysis of the logistic sigmoid function $\sigma_\beta(x) = (1+e^{-\beta x})^{-1}$, which governs pairwise comparison probabilities. The non-convexity of $\sigma_\beta$ makes standard linearization arguments insufficient. For the \C and \B upper bounds, we exploit the concavity of $\sigma_\beta$ on $[0,\infty)$ and its convexity on $(-\infty,0]$ to obtain tight approximations of expected vote shares. For \Plup, which requires bounding top-choice probabilities rather than pairwise comparisons, we construct convex surrogate functions of $\s$ that enable the analysis. 

Our lower bound constructions carefully balance voter utilities across multiple voter types to ensure the voting rule selects a low-welfare candidate with high probability in the large-election limit. Each construction illuminates a weakness pattern for the voting rule. \Plu has the spoiler effect, \R lets everyone pick their favorite candidate, which may be very bad overall, \PluV and \PPV give an extreme penalty for being at the end of a rank list, and \C discards the victory margins in pairwise comparisons.

\paragraph{Comparison with Classical Distortion.}
The classical results primarily~\citep{procaccia2006distortion, caragiannis2011voting, gkatzelis2023best} use the \emph{unit-sum} model, where $\sum_j u_{i,j} = 1$ for each voter. We work with \emph{bounded-range} utilities $u_{i,j} \in [0,1]$, which is natural for the PL model. Since unit-sum is a special case of bounded-range, our upper bounds directly apply to the unit-sum setting as well.

\begin{table}[t]
\centering
\setlength{\tabcolsep}{4pt}
\begin{tabular}{lcc}
\toprule
\textbf{Voting Rule} & \textbf{Upper Bound} & \textbf{Lower Bound} \\
\midrule
\B & $\beta\,\frac{1+e^{-\beta}}{1-e^{-\beta}}$ (Thm~\ref{thm:borda_ub}) & $(1-o(1))\beta$~(\citet{golz2025distortion}) \\[4pt]
\C & $\beta\,\frac{1+e^{-\beta}}{1-e^{-\beta}}$ (Thm~\ref{thm:copeland_ub}) & $(1-\epsilon)\beta$ (Thm~\ref{thm:copeland_lb}) \\[4pt]
\Plu & $\min\!\big(\frac{e^{2\beta}}{\beta},\, \frac{m e^{\beta}}{\ln{(m-1)}+2}+1\big)$ (Thm~\ref{thm:plurality_ub}) & $\min\!\big(\frac{e^{\beta}}{2+\epsilon}-1,\, \frac{m}{2+\epsilon}-1\big)$ (Thm~\ref{thm:plurality_lb}) \\[4pt]
\PluV & -- & $\frac{\beta \ln m}{6}$ (Thm~\ref{thm:pluV_lb}) \\[4pt]
\PPV & -- & $\frac{\beta \ln m}{6}$ (Cor~\ref{corr:ppv}) \\[4pt]
\R & $m e^\beta$ (Thm~\ref{thm:rd_ub}) & $(1-\epsilon)m$ (Thm~\ref{thm:rd_lb}) \\
\bottomrule
\end{tabular}
\caption{Utilitarian distortion bounds under Plackett-Luce voting with parameter $\beta$, as $n \to \infty$. Utilities lie in $[0,1]$. For \Plup, \C and \Rr, $\epsilon$ is any constant $>0$.}
\label{tab:results}
\end{table}

\subsection{Related Work}
\label{sec:related}

\paragraph{Probabilistic Voting and Random Utility Models.}
In economics, probabilistic voting models~\citep{coughlin1992probabilistic, quinn1999voter, mckelvey2006theory} assume that voters choose among candidates based on noisy utilities, reflecting uncertainty about candidates' positions. In the preference learning literature, this perspective is captured by random utility models (RUMs), such as the Plackett-Luce and Bradley-Terry models~\citep{bradley1952rank, plackett1975analysis, azari2012random, parkes2012random, caragiannis2016noisy}. These models have been used for statistical estimation of social rankings \cite{xia2013designing} and for axiomatic characterizations of voting rules as maximum likelihood estimators \cite{truchon2008borda}.

\paragraph{Distortion of Voting Rules: Classical and Probabilistic}
Most early work on utilitarian distortion considers the unit-sum model where each voter's utilities sum to one~\citep{procaccia2006distortion, boutilier2015optimal, caragiannis2011voting}. A parallel line studies metric distortion, where voters and candidates lie in a shared metric space~\citep{anshelevich2015approximating, goel2017metric, charikar2024breaking}. \citet{gkatzelis2023best} designed voting rules attaining near-optimal distortion in both settings.

\citet{goyal2025metric} introduced the paradigm of studying distortion under probabilistic voting, focusing on metric distortion under the Plackett-Luce model. They showed that normatively appealing rules like \C achieve improved metric distortion bounds when accounting for probabilistic voter behavior, while rules like \R and \PluV perform worse. \citet{golz2025distortion} later studied this problem under utilitarian distortion. They proved an $O(\beta^2)$ upper bound and a $(1-o(1))\beta$ lower bound for \Bb, and established that the \ML rule achieves the optimal distortion of $\frac{\beta}{2}\frac{1+e^{-\beta}}{1-e^{-\beta}}$ among randomized rules satisfying PCLC.

Together with the results of \citet{goyal2025metric}, our work shows that \C is close to the ``best of both distortion worlds'' under probabilistic voting among deterministic rules.

A complementary line of work improves distortion by giving the mechanism more than just rankings, for example via cardinal queries~\citep{abramowitz2019awareness, amanatidis2021peeking}, threshold approvals~\citep{anshelevich2024improved}, or deliberation~\citep{goel2025metric}. In our model, voters still submit ordinal rankings, but we crucially exploit the structure of the probabilistic mapping from utilities to rankings.

\paragraph{Distortion and AI Alignment.}
The connection between social choice and AI alignment has been emphasized in several recent papers~\citep{conitzer2024social, mishra2023ai, sorensen2024roadmap}. RLHF~\citep{ouyang2022training, stiennon2020learning} is the dominant paradigm for aligning large language models; however, most methods do not account for the fact that humans may have fundamentally different preferences. This has led to growing interest in alignment under preference heterogeneity~\citep{chidambaram2024direct, siththaranjan2023distributional, wang2025mpo, zhong2024provable}. As argued by \citet{golz2025distortion}, the social choice community can offer insights for designing RLHF pipelines that learn representative policies with theoretical guarantees. The distortion framework is particularly relevant since it provides robust guarantees independent of the preference distribution.


\section{Preliminaries}
\label{sec:prelims}


Let $C$ be a set of $m$ candidates (or alternatives) and let $V$ be a set of $n$ voters. Each voter $i \in V$ has a utility vector $\mathbf{u}_i = (u_{i,1}, \ldots, u_{i,m})$, where $u_{i,j} \in [0,1]$ denotes the utility of voter $i$ for candidate $j$. We assume voters are drawn independently from a distribution $\mathcal{V}$ over utility vectors.

We denote the total utility of candidate $j$ by $\U_j = \sum_{i \in V} u_{i,j}$, and the expected utility by $\textsc{Util}_j = \E_{\mathbf{u}_{\cdot} \sim \mathcal{V}}[u_{\cdot,j}]$. All our results are stated for the limit $n \to \infty$, where, by the strong law of large numbers (SLLN), $\frac{\U_j}{n} \to \textsc{Util}_j$ almost surely.

\subsection{Probabilistic Voting: The Plackett-Luce Model}
\label{sec:pl}

Voters have cardinal utilities but submit ordinal rankings $\pi$ over $C$. We model the generation of rankings from utilities using a probabilistic voting model, capturing the idea that voters' expressed preferences are noisy reflections of their true utilities.

\paragraph{Random Utility Interpretation.}
A natural way to model noisy ranking behavior is through random utility models (RUMs). Suppose each voter $i$ perceives a noisy utility $\hat{u}_{i,j} = u_{i,j} + \varepsilon_{i,j}$ for each candidate $j$, where $\{\varepsilon_{i,j}\}$ are i.i.d.\ noise terms. The voter then ranks candidates in decreasing order of perceived utilities. When the noise terms follow a Gumbel (Type-I extreme value) distribution with scale $1/\beta$, this procedure induces exactly the Plackett-Luce distribution over rankings~\citep{luce2005individual, plackett1975analysis}. The parameter $\beta \geq 1$ controls the noise level: higher $\beta$ corresponds to less noise, so rankings more closely reflect the true utility ordering.

\begin{definition}[Plackett-Luce Model]
\label{def:pl}
Under the Plackett-Luce model with parameter $\beta$, a ranking is constructed sequentially: at each step, a candidate is selected from the remaining candidates with probability proportional to their ``strength'' $\eta_{i,j} = e^{\beta u_{i,j}}$. That is, voter $i$ selects their top choice $j$ with probability $\frac{e^{\beta u_{i,j}}}{\sum_{k \in C} e^{\beta u_{i,k}}}$, then selects their second choice from the remaining candidates according to the same rule, and so on.
\end{definition}

\paragraph{Pairwise Comparisons and Bradley-Terry.}
A key consequence of the Plackett-Luce model is that pairwise comparison probabilities take a simple form. For any pair of candidates $X$ and $Z$:
\begin{equation}
\label{eq:plackett-luce}
\Pr(X >_i Z) = \sigma_\beta(u_{i,X} - u_{i,Z}),
\end{equation}
where $\sigma_\beta(x) = \frac{1}{1 + e^{-\beta x}}$ is the logistic sigmoid function. This is precisely the Bradley-Terry model for pairwise comparisons~\citep{bradley1952rank}. Note that the pairwise probability depends only on the \emph{difference} in utilities, consistent with the additive noise interpretation.

The probability that candidate $j$ is ranked first by voter $i$ is:
\begin{equation}
\label{eq:top}
\Pr(\top_i = j) = \frac{e^{\beta u_{i,j}}}{\sum_{k \in C} e^{\beta u_{i,k}}}.
\end{equation}

Let $\bot_i$ be the candidate ranked last by voter $i$. $\Pr(\bot_i = j)$ is determined by the sequential construction: $j$ is ranked last iff another candidate is selected at each step of ranking generation. 

The function $\sigma_\beta$ has several properties that we exploit in our proofs:
\begin{enumerate}
    \item $\sigma_\beta(x)$ is increasing for all $x$, concave for $x \geq 0$, and convex for $x \leq 0$.
    \item $\sigma_\beta'(x) = \beta \sigma_\beta(x)(1 - \sigma_\beta(x))$, with maximum value $\frac{\beta}{4}$ attained at $x = 0$.
    \item $\sigma_\beta(0) = \frac{1}{2}$ and $\sigma_\beta(x) + \sigma_\beta(-x) = 1$ for all $x$.
\end{enumerate}

\subsection{Voting Rules}
\label{sec:voting_rules}
A voting rule aggregates the rankings into a winning candidate (or a distribution over candidates).

\begin{definition}[Voting Rule]
A voting rule $f: \Pi^n \to \Delta(C)$ takes a profile of $n$ rankings and outputs a probability distribution over candidates. A rule is \emph{deterministic} if it always outputs a point mass, i.e., selects a single winner.
\end{definition}

We study the following voting rules. Let $s_{j,j'} = \frac{1}{n}|\{i \in V : j >_i j'\}|$ denote the fraction of voters who rank $j$ above $j'$, and let $t_j = \frac{1}{n}|\{i \in V : \top_i = j\}|$ denote the fraction of voters who rank $j$ first.

\begin{definition}[Random Dictator]
\label{def:rd}
Select a voter uniformly at random and output their top choice. That is, $\R(\pi^n) = \mathbf{p}$ where $p_j = t_j$.
\end{definition}

\begin{definition}[Plurality]
\label{def:plurality}
Select the candidate who is the top choice of the most voters:

$\Plu(\pi^n) = \argmax_{j \in C} t_j$.
\end{definition}

\begin{definition}[Copeland]
\label{def:copeland}
Let $\triangleright$ be a fixed strict ordering over the candidates used for tie-breaking. Select the candidate who wins the most pairwise comparisons, i.e, has the highest Copeland score:

$$
\C(\pi^n) = \argmax_{j \in C} \sum_{j' \in C \setminus \{j\}} \left[ \mathbbm{1}\left(s_{j,j'} > \frac{1}{2}\right) + \mathbbm{1}\left(s_{j,j'} = \frac{1}{2} \land j \triangleright j'\right) \right].
$$
\end{definition}

\begin{definition}[Borda]
\label{def:borda}
The Borda score of candidate $j$ is $\textsc{Borda}(j) = n \sum_{j' \neq j} s_{j,j'}$. The Borda rule selects the candidate with the highest Borda score: $\B(\pi^n) = \argmax_{j \in C} \textsc{Borda}(j)$.
\end{definition}

\begin{definition}[Plurality Veto]
\label{def:plurality-veto}
The rule proceeds in two phases. In the first phase, each voter places a token on their top choice. In the second phase, voters sequentially (in an arbitrary but fixed order) remove a token from the candidate ranked last among those with tokens remaining. The candidate with a token remaining after $n-1$ removals wins.
\end{definition}

\begin{definition}[Pruned Plurality Veto]
\label{def:pruned-plurality-veto}
For a constant $\alpha > 0$, first restrict to candidates with at least $\frac{\alpha n}{(6+\alpha)m}$ top-choice votes. Then run \PluV on the restricted candidate set.
\end{definition}

\subsection{Utilitarian Distortion}

We measure the quality of a voting rule by its \emph{utilitarian distortion}, which compares the social welfare of the chosen candidate to that of the optimal candidate.

\begin{definition}[Utilitarian Distortion under Probabilistic Voting]
\label{def:distortion}
The utilitarian distortion of a voting rule $f$ under probabilistic voting model $\mathcal{P}$ is:
\begin{equation}
\D_{\mathcal{P}}(f) = \sup_{\mathcal{V}} \limsup_{n \to \infty} \frac{\max_{j \in C}\textsc{Util}_j}{\textsc{Util}(f)},
\end{equation}
where $\textsc{Util}(f)=\E_{\mathbf{u} \sim \mathcal{V}^n,\, \pi^n \sim \mathcal{P}(\mathbf{u})} \left[\frac{\U(f(\pi^n))}{n}\right]$. Here, the expectation is over the randomness in the input rankings and the potential randomness in the voting rule $f$. 
\end{definition}

A distortion of $\alpha$ guarantees that the chosen candidate achieves at least a $\frac{1}{\alpha}$ fraction of the optimal social welfare, in the worst case over all voter distributions. We denote the socially optimal candidate (i.e., $\argmax_j \textsc{Util}_j$) by $B$. We drop the subscript $\mathcal{P}$ when clear from context.

\subsection{Population Quantities and Large Elections}\label{sec:largeElec}

Since our analysis is in the limit $n \to \infty$, we often work with population quantities. For a distribution $\mathcal{V}$ and parameter $\beta$, define the population pairwise probability and top-choice probability:
\begin{align}
p_{j,j'} &:= \Pr_{\mathbf{u}_{\cdot} \sim \mathcal{V},\, \pi \sim \mathcal{P}(\mathbf{u}_{\cdot})}(j > j') = \E_{\mathbf{u}_{\cdot} \sim \mathcal{V}}[\sigma_\beta(u_{\cdot,j} - u_{\cdot,j'})], \\
p_j^{\textsc{top}} &:= \Pr_{\mathbf{u}_{\cdot} \sim \mathcal{V},\, \pi \sim \mathcal{P}(\mathbf{u}_{\cdot})}(\top = j) = \E_{\mathbf{u}_{\cdot} \sim \mathcal{V}}\left[\frac{e^{\beta u_{\cdot,j}}}{\sum_{k \in C} e^{\beta u_{\cdot,k}}}\right].
\end{align}

By standard concentration arguments (SLLN), the empirical quantities $s_{j,j'}$, $t_j$, and $\frac{\U_j}{n}$ converge almost surely to their population counterparts $p_{j,j'}$, $p_j^{\textsc{top}}$, and $\textsc{Util}_j$ as $n \to \infty$. 


\subsection{Connection to Smoothed Analysis}
\label{sec:smooth}
In smoothed analysis~\citep{spielman2004smoothed}, an adversary selects a worst-case instance, and then ``nature'' perturbs it before processing. In social choice, this paradigm has been used to study the satisfiability of axioms that fail in the worst case~\citep{xia2020smoothed, flanigan2023smoothed, xia2023semi} and the runtime of voting rules~\citep{liu2022semi, xia2021smoothed, xia2022beyond}. Our work can be viewed as a smoothed analysis of utilitarian distortion guarantees.

Recall from \S\ref{sec:pl} that Plackett-Luce rankings arise from additive Gumbel perturbations: each voter $i$ perceives noisy utilities $\hat{u}_{i,j} = u_{i,j} + \varepsilon_{i,j}$, where $\{\varepsilon_{i,j}\}$ are i.i.d.\ Gumbel with scale $1/\beta$, and ranks candidates by decreasing $\hat{u}_{i,j}$. In smoothed analysis terms, the adversary chooses the utility profile $\{u_{i,j}\}$, and nature adds the Gumbel noise before rankings are generated.

Our setting differs from classical smoothed analysis in that distortion is evaluated with respect to the underlying (unperturbed) utilities $u$, while randomness enters only through the mapping to rankings. On the other hand, in smoothed analysis, it is done on the perturbed input. Nevertheless, there is a close connection. If we evaluate welfare on the noisy utilities $\hat{u}$, then for each candidate,
\[
\sum_{i \in V} \hat{u}_{i,j} = \U_j + \sum_{i \in V} \varepsilon_{i,j}.
\]
With a centered choice $\E[\varepsilon_{i,j}] = 0$, the strong law of large numbers gives $\frac{1}{n} \sum_{i \in V} \varepsilon_{i,j} \to 0$ almost surely. Hence, in large elections, the social welfare of candidate $j$ under $\hat{u}$ concentrates around $\U_j$. Consequently, whenever the denominator is bounded away from zero, ratios of social welfares under $\hat{u}$ converge to those under $u$ as $n \to \infty$.


Before we discuss our main results, we recall the lower bound for randomized rules that satisfy the probabilistic Condorcet loser criterion \citet{golz2025distortion}.
 It says that some amount of distortion is unavoidable and this bound grows linearly in the parameter $\beta$. That is, lower-noise regimes  have more unavoidable loss in how much social welfare can be guaranteed from observing mere ordinal preferences. 
For comparison, in the classical model with unit-sum utilities and deterministic voting, the lower bound for any deterministic rule is $\Omega(m^2)$, while for randomized rules it is $\Omega(\sqrt{m})$.
\begin{definition} \label{def:pclc}[Probabilistic Condorcet Loser Criterion (PCLC)]
   A rule satisfies PCLC if the probability that a Condorcet loser wins is less than $1/m$. A Condorcet loser is a candidate who loses to all other candidates in pair-wise comparisons.
\end{definition}
 Many natural voting rules, e.g., \ML and \Cc, satisfy PCLC. \Rr is an example of those that do not.
\begin{lemma}[Lower Bound~\citep{golz2025distortion}] \label{lem:general-lb}
For any social choice function $f$, which satisfies the Probabilistic Condorcet Loser Criterion, 
    $\D(f) \geq \frac{\beta}{2} \frac{1+e^{-\beta}}{1-e^{-\beta}}.$
\end{lemma}
In particular, \textsc{MaximalLotteries} matches this bound. Getting a higher lower bound for all deterministic rules is an interesting open problem.

\section{\R}
We begin by analyzing rules that do not fully exploit ranking information.
The \R rule (Definition~\ref{def:rd}) selects a candidate with probability proportional to their top-choice frequency. In the classical utilitarian distortion model, \R achieves $\Theta(m^{\frac{3}{2}})$ distortion \cite{gkatzelis2023best}. Under the Plackett-Luce model, we show that while the distortion is capped at $me^\beta$, there exist instances where it scales linearly with $m$.

\begin{theorem} \label{thm:rd_ub}
The utilitarian distortion of \R under the PL model is at most $m e^\beta$.
\end{theorem}

\begin{proof}
The probability $p_j$ that candidate $j$ is selected by \R is given by the population top-choice probability: $p_j = \mathbb{E}_{\mathbf{u} \sim \mathcal{V}} \left[ \frac{e^{\beta u_j}}{\sum_{k \in C} e^{\beta u_k}} \right]$. The expected social utility of the rule is $\textsc{Util}(\R) = \sum_{j \in C} p_j \textsc{Util}_j$.
Since all terms are non-negative, we lower bound the expected utility by the contribution of the optimal candidate $B$:
\begin{equation}
    \sum_{j \in C} p_j \textsc{Util}_j \geq p_B \textsc{Util}_B.
\end{equation}
Thus, the distortion is bounded by $\D(\R) \leq \frac{\textsc{Util}_B}{p_B \textsc{Util}_B} = \frac{1}{p_B}$. Using the bounded-range utility assumption $u_{i,j} \in [0,1]$, we have $e^{\beta u_B} \geq e^0 = 1$ and $\sum_{k \in C} e^{\beta u_k} \leq m e^\beta$. This yields:
\begin{equation}
    p_B = \mathbb{E} \left[ \frac{e^{\beta u_B}}{\sum_{k \in C} e^{\beta u_k}} \right] \geq \frac{1}{m e^\beta}.
\end{equation}
It follows that $\D(\R) \leq m e^\beta$.
\end{proof}

\begin{theorem} \label{thm:rd_lb}
The utilitarian distortion of \R under the Plackett-Luce model is at least $(1-\varepsilon)m$ for any constant $\varepsilon \in (0, \frac{m-2}{m-1})$, as $\beta \to \infty$.
\end{theorem}

\begin{proof}
We construct an instance with $m$ candidates: $B$ and $\{1, 2, \ldots, m-1\}$. Let $\mathcal{V}$ be a distribution over $m-1$ voter types, each occurring with equal probability $\frac{1}{m-1}$. A voter of type $k \in \{1, \ldots, m-1\}$ has utilities:
\[
u_{i,B} = 1 - \varepsilon, \qquad u_{i,k} = 1, \qquad u_{i,j} = 0 \text{ for } j \notin \{B, k\}.
\]

The expected utilities are $\textsc{Util}_B = 1 - \varepsilon$ and $\textsc{Util}_k = \frac{1}{m-1}$ for each $k \in \{1, \ldots, m-1\}$. Since $1 - \varepsilon > \frac{1}{m-1}$ by assumption, candidate $B$ is socially optimal.

For a voter of type $k$, the Plackett-Luce normalizer is
\[
Z = e^{\beta \cdot 1} + e^{\beta(1-\varepsilon)} + (m-2) \cdot e^0 = e^\beta + e^{\beta(1-\varepsilon)} + m - 2.
\]
The probability that this voter selects $B$ as their top choice is
\[
q_B = \frac{e^{\beta(1-\varepsilon)}}{Z} = \frac{1}{e^{\beta\varepsilon} + 1 + (m-2)e^{-\beta(1-\varepsilon)}},
\]
which is the same for all voter types. The probability of selecting candidate $k$ is $\frac{e^\beta}{Z}$, and the probability of selecting any $j \notin \{B, k\}$ is $\frac{1}{Z}$.
Since \R selects a uniformly random voter and outputs their top choice, the expected utility is
\[
\textsc{Util}(\R) = q_B \cdot \textsc{Util}_B + (1 - q_B) \cdot \frac{1}{m-1} = q_B(1-\varepsilon) + \frac{1-q_B}{m-1}.
\]

The distortion is therefore
\[
\D(\R) = \frac{\textsc{Util}_B}{\textsc{Util}(\R)} = \frac{1 - \varepsilon}{q_B(1-\varepsilon) + (1-q_B)/(m-1)}.
\]

As $\beta \to \infty$, we have $e^{\beta\varepsilon} \to \infty$, so $q_B \to 0$. Therefore,
\[
\D(\R) \to \frac{1-\varepsilon}{1/(m-1)} = (1-\varepsilon)(m-1) \leq (1-\varepsilon)m. \qedhere
\] 
\end{proof}

\section{\Plu}
We now study the \Plu rule, which selects the candidate with the highest number of top votes as the winner. In the classical utilitarian distortion model with unit-sum utilities, \Plu is order-wise optimal among deterministic rules with distortion $\Theta(m^2)$~\citep{caragiannis2011voting}, despite discarding most of the information in the rankings. 
 When the randomness is relatively low $(e^{\beta} \geq m),$ the distortion is at least linear in $m.$  When randomness is high $(e^{\beta} \leq m),$ the distortion is at least exponential in $\beta$. As we will see later, this is much worse than \C and \Bb, thereby aligning with conventional intuitions based on axiomatic properties of voting rules.

\begin{theorem}\label{thm:plurality_lb}
For $m \geq 4$ and $\beta \geq \ln 4$ and $\epsilon > 0$, the utilitarian distortion of \Plu under the Plackett-Luce model is at least $\min\left(\frac{e^{\beta}}{2+\epsilon} -1, \frac{m}{2 + \epsilon}-1\right).$ 
\end{theorem}

This lower bound arises from the \emph{spoiler effect}: many equally good candidates split votes, allowing a candidate with concentrated but low total support to win. This phenomenon also causes high metric distortion in both classical~\citep{anshelevich2021distortion} and probabilistic settings~\citep{goyal2025metric}.

\begin{proof}
Let $\mathcal{V}$ be a distribution over two voter types with fractions $\gamma$ and $1-\gamma$:
\begin{center}
\begin{tabular}{lcc}
& $u_W$ & $u_j$ for $j \neq W$ \\[2pt]
Type 1 (fraction $\gamma$): & $1$ & $0$ \\
Type 2 (fraction $1-\gamma$): & $0$ & $1$
\end{tabular}
\end{center}
The social utilities are $\textsc{Util}_W = \gamma$ and $\textsc{Util}_j = 1 - \gamma$ for all $j \neq W$. Thus any $j \neq W$ is optimal, and if $W$ wins, the distortion is $\frac{1-\gamma}{\gamma}$.

Since all candidates $j \neq W$ have identical utilities for all voters, they split votes evenly, each receiving an expected vote share of $\frac{1 - s_W}{m-1}$, where $s_W$ is the expected vote share of $W$. By standard concentration arguments, if $s_W > \frac{1}{m} + \epsilon$ for any constant $\epsilon > 0$, then $W$ receives strictly more votes than any other candidate with probability approaching $1$ as $n \to \infty$, so $W$ wins.

Type 1 voters rank $W$ first with probability $\Pr(\top_i = W) = \frac{e^\beta}{e^\beta + (m-1)}$. The expected vote share of $W$ is therefore at least $\gamma \cdot \frac{e^\beta}{e^\beta + (m-1)}$.

\textbf{Case 1: $e^\beta \leq m$.} Set $\gamma = \frac{2+\epsilon}{e^\beta}$ for small $\epsilon > 0$. Since $\frac{e^\beta}{e^\beta + (m-1)} > \frac{e^\beta}{2m}$, the expected vote share of $W$ is at least $\frac{2+\epsilon}{e^\beta} \cdot \frac{e^\beta}{2m} = \frac{2+\epsilon}{2m} > \frac{1}{m}$, so $W$ wins. The distortion is $\frac{1-\gamma}{\gamma} = \frac{e^\beta - (2+\epsilon)}{2+\epsilon}$.

\textbf{Case 2: $e^\beta > m$.} Set $\gamma = \frac{2+\epsilon}{m}$ for small $\epsilon > 0$. Since $\frac{e^\beta}{e^\beta + (m-1)} > \frac{1}{2}$, the expected vote share of $W$ is at least $\frac{2+\epsilon}{m} \cdot \frac{1}{2} = \frac{2+\epsilon}{2m} > \frac{1}{m}$, so $W$ wins. The distortion is $\frac{1-\gamma}{\gamma} = \frac{m - (2+\epsilon)}{2+\epsilon}$.

Combining these cases yields distortion at least $\min\left(\frac{e^\beta}{2+\epsilon} - 1, \frac{m}{2+\epsilon} - 1\right)$.
\end{proof}

However, \Plu has a bounded distortion. We have the following upper bound:

\begin{theorem}\label{thm:plurality_ub}
The utilitarian distortion of \Plu under the Plackett-Luce model is upper bounded by $\min\left(\frac{me^{\beta}}{\ln{(m-1)}+2} + 1, \frac{e^{2\beta}}{\beta}\right)$.
\end{theorem}

\begin{proof}
We work in the large-$n$ limit. By standard concentration (\S\ref{sec:largeElec}), the empirical top-choice frequencies $t_j$ and normalized social utilities $\U_j/n$ converge almost surely to their population counterparts $p_j^{\textsc{top}}$ and $\textsc{Util}_j$ respectively. Consequently, any candidate $W$ selected by \Plu with non-vanishing probability must be a maximizer of $p_j^{\textsc{top}}$ over $j \in C$, and in particular must satisfy $p_W^{\textsc{top}} \geq p_B^{\textsc{top}}$, where $B$ is the socially optimal candidate. Since the expected utility of the rule is a convex combination over such maximizers, it suffices to show that $\textsc{Util}_B / \textsc{Util}_W$ is bounded by the claimed quantity for every such $W$.

For a generic draw $\mathbf{u} \sim \mathcal{V}$, write $w := u_W$, $b := u_B$. The constraint $p_W^{\textsc{top}} \geq p_B^{\textsc{top}}$ reads:
\begin{equation}\label{eq:plu_pop_constraint}
\E\left[\frac{e^{\beta w} - e^{\beta b}}{\sum_{j \in C} e^{\beta u_j}}\right] \geq 0.
\end{equation}
Define $\mathcal{S}^+ := \{w \geq b\}$ and $\mathcal{S}^- := \{w < b\}$. Rearranging:
\begin{equation}\label{eq:plu_partition}
\E\left[\frac{(e^{\beta w} - e^{\beta b})\,\mathbbm{1}_{\mathcal{S}^+}}{\sum_{j \in C} e^{\beta u_j}}\right] \geq \E\left[\frac{(e^{\beta b} - e^{\beta w})\,\mathbbm{1}_{\mathcal{S}^-}}{\sum_{j \in C} e^{\beta u_j}}\right].
\end{equation}
We relax this inequality by weakening the LHS and strengthening the RHS. On $\mathcal{S}^+$, we minimize the denominator: $\sum_j e^{\beta u_j} \geq e^{\beta w} + (m-1) \cdot e^0 = e^{\beta w} + m - 1$, and use $e^{\beta b} \geq 1$ (since $b \geq 0$) to replace the numerator $e^{\beta w} - e^{\beta b}$ by the larger $e^{\beta w} - 1$. On $\mathcal{S}^-$, we maximize the denominator: $\sum_j e^{\beta u_j} \leq m \cdot e^{\beta}$. This yields:
\begin{equation}\label{eq:plu_relaxed}
\E\left[\frac{(e^{\beta w} - 1)\,\mathbbm{1}_{\mathcal{S}^+}}{e^{\beta w} + m - 1}\right] \geq \E\left[\frac{(e^{\beta b} - e^{\beta w})\,\mathbbm{1}_{\mathcal{S}^-}}{me^{\beta}}\right].
\end{equation}

\paragraph{Lower bounding the RHS}
By the mean value theorem, $e^{\beta b} - e^{\beta w} = \beta e^{\beta c}(b - w)$ for some $c \in (w, b)$. Since $c \geq w \geq 0$, we have $e^{\beta c} \geq 1$, so $e^{\beta b} - e^{\beta w} \geq \beta(b - w)$ pointwise on $\mathcal{S}^-$. Applying this to the RHS of~\eqref{eq:plu_relaxed}:
\begin{equation}\label{eq:plu_rhs_mvt}
\E\left[\frac{(e^{\beta b} - e^{\beta w})\,\mathbbm{1}_{\mathcal{S}^-}}{me^{\beta}}\right] \geq \frac{\beta}{me^{\beta}}\,\E\bigl[(b - w)\,\mathbbm{1}_{\mathcal{S}^-}\bigr].
\end{equation}
To bound $\E[(b-w)\,\mathbbm{1}_{\mathcal{S}^-}]$, note that $\E[b\,\mathbbm{1}_{\mathcal{S}^-}] = \textsc{Util}_B - \E[b\,\mathbbm{1}_{\mathcal{S}^+}] \geq \textsc{Util}_B - \E[w\,\mathbbm{1}_{\mathcal{S}^+}]$, where the inequality uses $b \leq w$ on $\mathcal{S}^+$. Therefore:
\begin{equation}\label{eq:plu_util_diff}
\E[(b - w)\,\mathbbm{1}_{\mathcal{S}^-}] = \E[b\,\mathbbm{1}_{\mathcal{S}^-}] - \E[w\,\mathbbm{1}_{\mathcal{S}^-}] \geq \textsc{Util}_B - \E[w\,\mathbbm{1}_{\mathcal{S}^+}] - \E[w\,\mathbbm{1}_{\mathcal{S}^-}] = \textsc{Util}_B - \textsc{Util}_W.
\end{equation}
Combining~\eqref{eq:plu_relaxed},~\eqref{eq:plu_rhs_mvt}, and~\eqref{eq:plu_util_diff}:
\begin{equation}\label{eq:plu_core}
\E\left[\frac{(e^{\beta w} - 1)\,\mathbbm{1}_{\mathcal{S}^+}}{e^{\beta w} + m - 1}\right] \geq \frac{\beta(\textsc{Util}_B - \textsc{Util}_W)}{me^{\beta}}.
\end{equation}

\paragraph{Upper bounding the LHS}
Define $g(w) := \frac{e^{\beta w} - 1}{e^{\beta w} + m - 1}$ for $w \in [0,1]$. Note that $g(0) = 0$, $g$ is increasing, and the LHS of~\eqref{eq:plu_core} is at most $\E[g(w)]$. We consider two cases.

\emph{Case 1: $e^\beta \leq m - 1$.}
The second derivative of $g$ is:
\begin{equation}
g''(w) = \frac{\beta^2 m \, e^{\beta w}(m - 1 - e^{\beta w})}{(e^{\beta w} + m - 1)^3}.
\end{equation}
When $e^\beta \leq m - 1$, we have $e^{\beta w} \leq e^{\beta} \leq m - 1$ for all $w \in [0,1]$, so $g''(w) \geq 0$. Hence $g$ is convex on $[0,1]$. Since $g(0) = 0$, convexity implies $g(w) \leq w \cdot g(1)$ for all $w \in [0,1]$. Therefore:
\begin{equation}
\E[g(w)] \leq g(1) \cdot \E[w] = \frac{e^{\beta} - 1}{e^{\beta} + m - 1} \cdot \textsc{Util}_W.
\end{equation}
Substituting into~\eqref{eq:plu_core}:
\begin{equation}
\frac{e^{\beta} - 1}{e^{\beta} + m - 1} \cdot \textsc{Util}_W \geq \frac{\beta(\textsc{Util}_B - \textsc{Util}_W)}{me^{\beta}} \implies \textsc{Util}_B \leq \textsc{Util}_W \left(\frac{me^{\beta}(e^{\beta} - 1)}{\beta(e^{\beta} + m - 1)} + 1\right).
\end{equation}
Since $e^{\beta} + m - 1 \geq m$, the first term is at most $\frac{e^{\beta}(e^{\beta} - 1)}{\beta}$. Using $e^{\beta} - 1 \leq e^{\beta}$ and $\beta \leq e^{\beta}$ (for $\beta \geq 1$):
\begin{equation}
\frac{\textsc{Util}_B}{\textsc{Util}_W} \leq \frac{e^{2\beta}}{\beta}.
\end{equation}

\emph{Case 2: $e^\beta > m - 1$.}
We show that $g(w) \leq \frac{\beta w}{\ln(m-1) + 2}$ for all $w \in [0,1]$. Substituting $x = \beta w$, it suffices to prove:
\begin{equation}
\frac{e^x - 1}{e^x + m - 1} \leq \frac{x}{\ln(m-1) + 2} \quad \text{for all } x \geq 0.
\end{equation}
Define $h(x) = x(e^x + m - 1) - (e^x - 1)(\ln(m-1) + 2)$. We need $h(x) \geq 0$ for $x \geq 0$. Observe that $h(0) = 0$, and:
\begin{equation}
h'(x) = e^x + (m-1) + xe^x - e^x(\ln(m-1) + 2) = e^x(x - \ln(m-1) - 1) + (m-1).
\end{equation}
The second derivative is $h''(x) = e^x(x - \ln(m-1))$, which is negative for $x < \ln(m-1)$ and positive for $x > \ln(m-1)$. Thus $h'$ attains its minimum at $x = \ln(m-1)$:
\begin{equation}
h'(\ln(m-1)) = (m-1)(\ln(m-1) - \ln(m-1) - 1) + (m-1) = -(m-1) + (m-1) = 0.
\end{equation}
Since $h'(x) \geq 0$ for all $x \geq 0$ and $h(0) = 0$, we conclude $h(x) \geq 0$ for all $x \geq 0$.

Therefore $g(w) \leq \frac{\beta w}{\ln(m-1) + 2}$, and:
\begin{equation}
\E[g(w)] \leq \frac{\beta}{\ln(m-1) + 2} \cdot \E[w] = \frac{\beta \cdot \textsc{Util}_W}{\ln(m-1) + 2}.
\end{equation}
Substituting into~\eqref{eq:plu_core}, since the LHS of~\eqref{eq:plu_core} is at most $\E[g(w)]$:
\begin{equation}
\frac{\beta \cdot \textsc{Util}_W}{\ln(m-1) + 2} \geq \frac{\beta(\textsc{Util}_B - \textsc{Util}_W)}{me^{\beta}}.
\end{equation}
Dividing both sides by $\beta > 0$ and rearranging:
\begin{equation}
\frac{\textsc{Util}_B}{\textsc{Util}_W} \leq \frac{me^{\beta}}{\ln(m-1) + 2} + 1.
\end{equation}

Taking the minimum of both cases yields the theorem.
\end{proof}

\section{\PluV and \PPV}
In certain cases, \PluV (Definition~\ref{def:plurality-veto}) can be an antidote to the spoiler effect that affects \Plup. This reduces the metric distortion \cite{kizilkaya2022plurality}, and in fact, \PluV achieves the best possible metric distortion for deterministic voting rules. However, this is not enough under the Plackett-Luce model, since \PluV also discards most of the information from the rankings, and the noise in the votes can cause its metric distortion to grow with $m$ \cite{goyal2025metric}. For the utilitarian model also, we show that \PluV has distortion increasing with $m,$ albeit more slowly, with $\ln m.$ 
\begin{theorem}\label{thm:pluV_lb}
The utilitarian distortion of \PluV under the Plackett-Luce model is at least $\frac{\beta \ln m}{6}$ when $\beta \leq \frac{m}{3\ln m}$, $\beta \geq 2\ln \beta + 2 \ln \ln m$, and $m \geq 10$.
\end{theorem}

\begin{corollary}\label{corr:ppv}
The utilitarian distortion of \PPV under the Plackett-Luce model is at least $\frac{\beta \ln m}{6}$ under the same conditions.
\end{corollary}
These conditions are satisfied for a broad range of $\beta$ for large $m$. For example, when $m=400,$ the theorem holds at least for all $\beta \in [7.77,22.22].$ The lower bound of $\beta$ for the theorem to hold can be relaxed to a simpler condition $\beta \geq 3(1+\ln \ln m).$

\begin{proof}[Proof sketch]
We construct a distribution $\mathcal{V}$ over two voter types, each with probability $\frac{1}{2}$.

\emph{Type 1 voters (probability $\frac{1}{2}$):} $u_{i,B} = 0$ and $u_{i,j} = \frac{2}{\beta \ln m}$ for all $j \neq B$. 

\emph{Type 2 voters (probability $\frac{1}{2}$):} $u_{i,B} = \frac{1}{2}$. Additionally, Type 2 voters are partitioned into $\binom{m-1}{\lfloor \frac{m-1}{\beta \ln m} \rfloor}$ subtypes, each with probability $\frac{1}{2}\frac{1}{\binom{m-1}{\lfloor \frac{m-1}{\beta \ln m} \rfloor}}$. Each subtype corresponds to a subset $S_i \subset C \setminus \{B\}$ of size $\lfloor \frac{m-1}{\beta \ln m} \rfloor$ such that $u_{i,j} = 1$ for $j \in S_i$ and $u_{i,j} = 0$ for $j \notin S_i \cup \{B\}$. Each candidate appears in an equal number of subsets, up to rounding. Therefore, for each candidate $j \neq B,$ $u_{i,j}=1$ for at most a $\frac{1}{\beta \ln m}$ fraction of Type 2 voters. 

Under this construction, $\textsc{Util}_B = \frac{1}{4}$, while by symmetry every other candidate has social utility at most $\frac{3}{2\beta \ln m}$. Thus $B$ is optimal, and if $B$ loses, the distortion is at least $\frac{\beta \ln m}{6}$.

We show $B$ loses under \PluV by comparing top and bottom votes. For top votes, Type 1 voters rank $B$ first with probability $\frac{1}{1 + (m-1)e^{2/\ln m}} \leq \frac{1}{m}$, since $e^{2/\ln m} > 1$. Type 2 voters rank $B$ first with probability $\frac{e^{\beta/2}}{e^{\beta/2} + \lfloor \frac{m-1}{\beta \ln m} \rfloor e^\beta + m - 1 - \lfloor \frac{m-1}{\beta \ln m} \rfloor}$, which is at most $\frac{3 e^{-\beta/2} \beta \ln m}{2m}$ using $\beta \leq \frac{m}{3 \ln m}$. Combined with the condition $\beta \geq 2\ln\beta + 2\ln\ln m$, the total expected top votes for $B$ is at most $\frac{5n}{4m}$.

For bottom votes, we consider only Type 1 voters. Under Plackett-Luce, $B$ is ranked last when it is sequentially \emph{not} selected at each position. Since Type 1 voters have $u_{i,B} = 0$ and $u_{i,j} = \frac{2}{\beta \ln m}$ for $j \neq B$, the probability $B$ is ranked last is $\prod_{k=1}^{m-1} \frac{(m-k)e^{2/\ln m}}{(m-k)e^{2/\ln m} + 1} \geq \frac{1}{(em)^{e^{-2/\ln m}}}$. Thus the expected bottom votes from Type 1 voters alone is at least $\frac{n}{2(em)^{e^{-2/\ln m}}}$.

For $m \geq 10$, one can verify that $\frac{1}{(em)^{e^{-2/\ln m}}} > \frac{5}{2m}$, so $B$ receives more bottom votes than top votes. In \PluVv, this ensures $B$'s tokens are removed before the final round, so $B$ loses.

For \PPVv, we verify that all candidates (including $B$) receive at least $\frac{n}{2em}$ top votes in expectation, which exceeds the pruning threshold $\frac{\alpha n}{(6+\alpha)m}$ for $\alpha \leq 1.3$. Thus no candidates are pruned, and the same analysis applies.
The full calculations appear in Appendix~\ref{app:pluv}.
\end{proof}

\section{\Cc}

So far, we have seen that rules based primarily on top-choice information (\Rr, \Plup, and \PluVv) have distortion growing with $m$ or exponentially in $\beta$. We now turn to tournament-based rules, which aggregate the full pairwise comparison information implicit in the rankings. These rules are particularly relevant for AI alignment, where preferences are often elicited through pairwise comparisons. We show that such rules achieve distortion independent of $m$.
We first establish two useful lemmas.

\begin{lemma}[Two-Candidate Upper Bound]\label{lem:2cand}
    If $p_{X,Z} \geq \frac{1}{2},$ we have $\frac{\U_Z}{\U_X} \leq \frac{\beta}{2} \frac{1+e^{-\beta}}{1-e^{-\beta}}.$
\end{lemma}
\begin{proof}
    This follows from the fact that the distortion of \textsc{MaximalLotteries} is at most $\frac{\beta}{2} \frac{1+e^{-\beta}}{1-e^{-\beta}}$ \citep{golz2025distortion}. When $m=2$, $X$ is a \textsc{MaximalLotteries} winner if $p_{X,Z} \geq \frac{1}{2}.$ The proof follows from the fact that $\frac{\textsc{Util}_Z}{\textsc{Util}_X}=\frac{\U_Z}{\U_X}$ as $n\to \infty$.
\end{proof}

This lemma implies that for the case of two candidates, any voting rule that picks the candidate with more votes as the winner cannot have a distortion worse than $\frac{\beta}{2} \frac{1+e^{-\beta}}{1-e^{-\beta}}.$ We will use this result as a base case in our proofs for the upper bounds of \C and \Bb.

The following lemma is technical and is needed to circumvent the non-convexity of $\s$ in analysis. Its purpose is to lower-bound the gain in vote fraction against a fixed opponent with an increase in social utility. In our proofs, we will use it to lower-bound the vote fraction of the socially optimal candidate $B$ against an arbitrary opponent.
\begin{lemma}\label{lem:linearization}
   Let  $u_{i,X} = x_i$ and $u_{i,Z} = z_i$ be the utilities of voter $i$ from candidates $X$ and $Z$ respectively. We have the following inequality:
   \begin{equation}
       \sum_{i \in V}  \s(x_i - z_i)  - \s(- z_i) \geq \frac{(\U_X - \U_Z)}{2}\frac{1-e^{-\beta}}{1+e^{-\beta}}.
   \end{equation}
\end{lemma}
\begin{proof}
The result holds trivially when $U_X < U_Z$ since the left-hand side is non-negative for $x_i \geq 0$. 
When $U_X \geq U_Z,$ we have the following.
    \begin{align}
         \sum_{i \in V}  \s(x_i - z_i)  - \s(- z_i) &= \sum_{i:x_i \geq z_i}  \s(x_i - z_i)  - \s(- z_i) + \sum_{i : x_i < z_i}  \s(x_i - z_i)  - \s(- z_i) \\
         &\geq \sum_{i : x_i \geq z_i}  \s(x_i - z_i)  - \s(- z_i) \\
         &\geq \sum_{i : x_i \geq z_i}  \s(x_i - z_i)  - \s(0).
         \intertext{Let $t_i$ denote $x_i - z_i$. Since $\s(t_i)$ is concave increasing in $t_i \geq 0,$ we have $\s(t_i) \geq \s(0) + t_i (\s(1)-\s(0)).$ Therefore, }
        \sum_{i \in V} \s(x_i - z_i)  - \s(- z_i) &\geq \sum_{i : x_i \geq z_i}  (x_i - z_i)(\s(1)-\s(0))
         \intertext{We know that $\sum\limits_{i: x_i \geq z_i}  x_i - z_i \geq \sum\limits_{i \in V}  x_i - z_i  = \U_X - \U_Z$. Therefore,}
        \sum_{i \in V} \s(x_i - z_i)  - \s(- z_i) &\geq (\U_X - \U_Z)(\s(1) - \s(0)) = \frac{(\U_X - \U_Z)}{2}\frac{1-e^{-\beta}}{1+e^{-\beta}}.
    \end{align}

This completes the proof.
\end{proof}

We are now ready to discuss the bounds on the distortion of \Cc. We first give the upper bound, and then a lower bound which is arbitrarily close to the upper bound when $\beta \rightarrow \infty$. 

\begin{theorem} \label{thm:copeland_ub}
   The utilitarian distortion of \C under the PL model is at most $\beta \frac{1+e^{-\beta}}{1-e^{-\beta}}.$
\end{theorem}
\begin{proof}
 Let $W$ be a candidate selected by \C with non-zero probability as $n \to \infty$. Since $W$ is a winner, their Copeland score must be at least as high as that of the socially optimal candidate $B$. This implies that at least one of the following conditions must hold:
\begin{enumerate}
    \item $p_{W,B} \geq \frac{1}{2}$.
    \item $\exists Y \in C\setminus\{W,B\}$ such that $p_{Y,B} \geq \frac{1}{2}$ and $p_{W,Y} \geq \frac{1}{2}.$
\end{enumerate}
%

For case 1, the upper bound follows from Lemma~\ref{lem:2cand}.
We now study case 2.

As shorthand, we use $w_i, y_i,$ and $b_i$ to denote $u_{i,W}, u_{i,Y},$ and $u_{i,B}$ respectively. Given the SLLN, the maximum possible distortion can be encoded as the optimal value of the following program.

\begin{equation}  \label{opt:copeland}
\optimizer_{\textsc{Cop}} =
\left\{
\begin{array}{ll}
\displaystyle
\text{maximize} &  \frac{\sum_{i \in V} b_i}{\sum_{i \in V} w_i}  \\[0.5em]
\text{subject to} &
\displaystyle \sum_{i \in V} \s(w_i - y_i) \geq \frac{n}{2}\\ 
& \displaystyle \sum_{i \in V} \s(y_i - b_i) \geq \frac{n}{2}\\
& w_i \in [0,1] ~~ \forall i \in V\\  
& y_i \in [0,1] ~~ \forall i \in V \\ 
& b_i \in [0,1] ~~ \forall i \in V\\
\end{array}
\right.
\end{equation}

Since the slope of $\s$ is upper bounded by $\frac{\beta}{4}$: 
\begin{equation}
    \sum\limits_{i \in V} \s(w_i - y_i) \leq \sum\limits_{i \in V} \frac{\beta}{4} w_i + \s(- y_i) = \frac{\beta \U_W}{4} + \sum\limits_{i \in V} \s(- y_i).
\end{equation}
The constraint in $\optimizer_{\textsc{Cop}}$ is $\sum\limits_{i \in V} \s(w_i - y_i) \geq \frac{n}{2}.$  Therefore, 
\begin{equation} \label{eq:sig-lb}
  \sum\limits_{i \in V} \sigma_{\beta}(- y_i) \geq  \frac{n}{2} - \frac{\beta \U_W}{4}.
\end{equation}

Lemma~\ref{lem:linearization} implies the following: 
\begin{align}
    &\sum_{i \in V}    \s(b_i - y_i)  - \s(- y_i) \geq  \frac{(\U_B - \U_Y)}{2}\frac{1-e^{-\beta}}{1+e^{-\beta}} \\
    \implies & n- \sum_{i \in V} \s(y_i - b_i) -  \sum_{i \in V} \s(-y_i) \geq    \frac{(\U_B - \U_Y)}{2}\frac{1-e^{-\beta}}{1+e^{-\beta}} \\
    \implies &   \sum_{i \in V} \s(-y_i) \leq    n - \sum_{i \in V} \s(y_i - b_i) -  \frac{(\U_B - \U_Y)}{2}\frac{1-e^{-\beta}}{1+e^{-\beta}} \\
    \intertext{The constraint in $\optimizer_{\textsc{Cop}}$ is $\sum\limits_{i \in V} \s(y_i - b_i) \geq \frac{n}{2}$. Therefore, }
    &   \sum_{i \in V} \s(-y_i) \leq    \frac{n}{2} - \frac{(\U_B - \U_Y)}{2}\frac{1-e^{-\beta}}{1+e^{-\beta}}.  \label{eq:sig-ub}
\end{align}

Equations~\eqref{eq:sig-lb} and~\eqref{eq:sig-ub} imply the following:

\begin{align}
    & \frac{n}{2} -  \frac{(\U_B - \U_Y)}{2}\frac{1-e^{-\beta}}{1+e^{-\beta}} \geq  \frac{n}{2} - \frac{\beta \U_W}{4}\\
    \implies &  \frac{\beta \U_W}{2} \geq (\U_B - \U_Y)\frac{1-e^{-\beta}}{1+e^{-\beta}}. \label{eq:cop_uw_core}
    \intertext{If  $\U_Y \geq \frac{\U_B}{2},$ the result follows from Lemma~\ref{lem:2cand} since $p_{W,Y} \geq \frac{1}{2}$. If  $\U_Y < \frac{\U_B}{2},$ we have,}
    &\frac{\beta \U_W}{2} \geq \frac{\U_B}{2} \frac{1-e^{-\beta}}{1+e^{-\beta}}.
\end{align}

Finally, if $\textsc{Util}_B=0$ then the distortion is $1$. Otherwise, by the strong law,
$\U_B/n\to \textsc{Util}_B>0$ a.s. and $\U_W/n\to \textsc{Util}_W$ a.s.
Moreover, either $\U_Y\ge \U_B/2$ (in which case Lemma~\ref{lem:2cand} with $p_{W,Y}\ge \tfrac12$
implies $\U_W \ge c_\beta \U_B$) or $\U_Y<\U_B/2$ (in which case \eqref{eq:cop_uw_core}
implies $\U_W \ge c_\beta \U_B$), where $c_\beta:=\frac{1}{\beta}\frac{1-e^{-\beta}}{1+e^{-\beta}}>0$.
Hence $\textsc{Util}_W \ge c_\beta \textsc{Util}_B>0$, so the welfare ratio is well-defined and
\[
\frac{\textsc{Util}_B}{\textsc{Util}_W}
=\lim_{n\to\infty}\frac{\U_B/n}{\U_W/n}
=\lim_{n\to\infty}\frac{\U_B}{\U_W}.
\]
Hence the bound on $\U_B/\U_W$ implies the same bound on $\textsc{Util}_B/\textsc{Util}_W$.
\end{proof}

We now give a nearly tight lower bound for \C with just three candidates.

\begin{theorem} \label{thm:copeland_lb}
The utilitarian distortion of \C under the PL model is at least $(1-\epsilon)\beta$ for any constant $\epsilon>0$, as $\beta \to \infty$.
\end{theorem}

\begin{proof}
We construct an instance with three candidates $B, W, Y$ and a distribution $\mathcal{V}$ over three voter types occurring with fractions $p$, $q$, and $1-p-q$ respectively:
\begin{center}
\begin{tabular}{lccc}
& $u_B$ & $u_Y$ & $u_W$ \\[2pt]
Type I (fraction $p$): & $1 - \eta$ & $1$ & $0$ \\
Type II (fraction $q$): & $1$ & $0$ & $\delta/\beta$ \\
Type III (fraction $1-p-q$): & $0$ & $0$ & $\delta/\beta$
\end{tabular}
\end{center}
where $\eta, \delta \in (0,1)$ are constants to be chosen later. The tie-breaking order is $W \triangleright Y \triangleright B$, used both for pairwise ties and ties in overall Copeland scores.

The social utilities are $\textsc{Util}_B = p(1-\eta) + q$ and $\textsc{Util}_W = (1-p)\frac{\delta}{\beta}$. If $W$ wins, the distortion is:
\begin{equation}\label{eq:cop-distortion}
\D = \frac{\textsc{Util}_B}{\textsc{Util}_W} = \beta \cdot \frac{p(1-\eta) + q}{\delta(1-p)}.
\end{equation}

\paragraph{Ensuring $W$ wins.}
We choose $p$ and $q$ so that all three population pairwise margins are strictly above $\frac{1}{2}$. By concentration, as $n \to \infty$ each empirical margin $s_{j,j'}$ converges to $p_{j,j'}$, so all three exceed $\frac{1}{2}$ with high probability. Each candidate then wins exactly one pairwise comparison, giving Copeland scores $(1,1,1)$, and $W$ wins by $\triangleright$.

To achieve this, we set $p_{W,Y} = \frac{1}{2} + \epsilon'$ and $p_{Y,B} = \frac{1}{2} + \epsilon'$, where $\epsilon' = 1/\beta > 0$. Computing:
\begin{align*}
p_{W,Y} &= p \cdot \sigma_\beta(-1) + (1-p) \cdot \sigma_\beta(\delta/\beta), \\
p_{Y,B} &= p \cdot \sigma_\beta(\eta) + q \cdot \sigma_\beta(-1) + (1-p-q) \cdot \tfrac{1}{2}.
\end{align*}

\begin{equation}\label{eq:p-value}
\text{Solving $p_{W,Y} = \frac{1}{2} + \epsilon'$ for $p$: } \qquad  p = \frac{\sigma_\beta(\delta/\beta) - \frac{1}{2} - \epsilon'}{\sigma_\beta(\delta/\beta) - \sigma_\beta(-1)}.
\end{equation}
\begin{equation}\label{eq:q-value}
\text{Solving $p_{Y,B} = \frac{1}{2} + \epsilon'$ for $q$:} \qquad q = \frac{p(\sigma_\beta(\eta) - \frac{1}{2}) - \epsilon'}{\frac{1}{2} - \sigma_\beta(-1)}.
\end{equation}

It remains to verify that $p_{B,W} > \frac{1}{2}$. The third margin is:
\[
p_{B,W} = p \cdot \sigma_\beta(1-\eta) + q \cdot \sigma_\beta(1 - \delta/\beta) + (1-p-q)(1 - \sigma_\beta(\delta/\beta)).
\]
As $\beta \to \infty$, the terms $\sigma_\beta(1-\eta), \sigma_\beta(1-\delta/\beta) \to 1$ and $\sigma_\beta(-1) \to 0$, while $\sigma_\beta(\delta/\beta) \to s := \sigma_1(\delta) \in (\frac{1}{2}, 1)$. Since $\epsilon' \to 0$, the parameters $p, q$ converge to their $\epsilon' = 0$ limits (computed below), and $p_{B,W} \to 2p + (1-2p)(1-s) = 1 - s + 2ps$. Substituting $p \to (s - \frac{1}{2})/s$ gives 
$p_{B,W} \to 1 - s + 2s \cdot \frac{s - \frac{1}{2}}{s} = s > \frac{1}{2}$. Hence $p_{B,W} > \frac{1}{2}$ for all sufficiently large $\beta$.

For large $\beta$, $\sigma_\beta(\delta/\beta) - \frac{1}{2} \to s - \frac{1}{2} > 0$ while $\epsilon' \to 0$, so $p > 0$. Similarly $\sigma_\beta(\eta) \to 1$ ensures $q > 0$. Since $\sigma_\beta(\delta/\beta) \leq \sigma_\beta(1) < 1$ and $\sigma_\beta(\eta) - \frac{1}{2} \leq \frac{1}{2} - \sigma_\beta(-1)$, we have $p \leq \frac{1}{2}$ and $q \leq p$, so $p + q \leq 1$.

As $\beta \to \infty$, $\epsilon' \to 0$, so $p$ and $q$ converge to the $\epsilon' = 0$ limits. Using $\sigma_\beta(-1) \to 0$, $\sigma_\beta(\eta) \to 1$, and the Taylor expansion $\sigma_1(\delta) = \frac{1}{2} + \frac{\delta}{4}(1 - \frac{\delta^2}{12}) + O(\delta^5)$, let $\alpha = \delta(1 - \frac{\delta^2}{12})$ for brevity. Then:
\begin{align*}
p &\to \frac{\alpha/4}{\frac{1}{2} + \alpha/4} = \frac{\alpha}{2 + \alpha}, \qquad
q \to \frac{p \cdot \frac{1}{2}}{\frac{1}{2}} = p.
\end{align*}
Substituting into \eqref{eq:cop-distortion}:
\begin{align*}
\lim_{\beta \to \infty} \frac{\D}{\beta}
&= \frac{p(1-\eta) + q}{\delta(1-p)}
= \frac{\frac{\alpha}{2+\alpha}(1-\eta) + \frac{\alpha}{2+\alpha}}{\delta \cdot \frac{2}{2+\alpha}}
= \frac{\alpha(2-\eta)}{2\delta}
= (1 - \tfrac{\delta^2}{12})(1 - \tfrac{\eta}{2}).
\end{align*}

For any $\epsilon \in (0, \frac{1}{4})$, choosing $\eta = \epsilon$ and $\delta = \sqrt{3\epsilon}$ gives:
\(
(1 - \tfrac{\delta^2}{12})(1 - \tfrac{\eta}{2}) = (1 - \tfrac{\epsilon}{4})(1 - \tfrac{\epsilon}{2}) \geq 1 - \epsilon,
\)
hence $\D \geq (1-\epsilon)\beta$ for all sufficiently large $\beta$.
\end{proof}

\section{\Bb}

We now discuss the bounds for \Bb. Perhaps surprisingly, our upper bound is the same as that of \Cc, i.e., $\beta \frac{1+e^{-\beta}}{1-e^{-\beta}}.$ This improves the $O(\beta^2)$ bound of prior work \cite{golz2025distortion}. \citet{golz2025distortion} gave a lower bound of $(1-o(1))\beta$ as $\beta \rightarrow \infty$\footnote{
\citet{golz2025distortion} proves the \((1-o(1))\beta\) lower bound in a slightly more general model in which the elicitation process induces \emph{non-uniform} frequencies of pairwise comparisons (equivalently, different candidates may be ``sampled'' with different probabilities when comparisons are generated). 
Although our baseline model draws full rankings from Plackett--Luce (which fixes these frequencies), their lower-bound construction can be simulated in our setting via \emph{cloning}: by replacing a candidate that is sampled with higher probability by identical clones (and splitting its sampling mass among the clones), we can reproduce the same effective comparison weights while keeping the underlying voting rule unchanged.
Conversely, our upper-bound analysis is robust to such non-uniform sampling as well: the proof only uses linear inequalities that hold \emph{pointwise} for each compared pair (via the Lipschitz/concavity properties of \(\sigma_\beta\)), and thus continues to hold after taking any convex combination of pairwise comparisons, i.e., under arbitrary sampling weights.}. Therefore, for \B as well, we now have almost close upper and lower bounds.
\begin{theorem}\label{thm:borda_ub}
   The utilitarian distortion of \B under the PL model is at most $\beta \frac{1+e^{-\beta}}{1-e^{-\beta}}.$ 
\end{theorem}
\begin{proof}

Let $W$ be any candidate selected by \B with non-zero probability in the large election limit. As established in \S\ref{sec:largeElec}, $W's$ expected Borda score must be maximal among all candidates, and specifically must be at least that of the socially optimal candidate $B$. Furthermore, since the sum of Borda scores is fixed, the maximum score must be at least the average score. Thus, the following conditions must hold for $W$.

\begin{enumerate}
    \item $\sum\limits_{j \in C\setminus \{W\}} p_{W,j} \geq \sum\limits_{j \in C\setminus \{B\}} p_{B,j}.$ That is, W must have a Borda score at least equal to that of B.
    \item $\sum\limits_{j \in C\setminus \{W\}} p_{W,j} \geq \frac{(m-1)}{2}.$ That is, W must have a higher Borda score than average.
\end{enumerate}

If $p_{W,B} \geq \frac{1}{2},$ the result follows from Lemma~\ref{lem:2cand}. If not, we can update the above conditions by subtracting $p_{W,B}$ from both sides.

\begin{enumerate}
    \item \label{cond:1} $\sum\limits_{j \in C\setminus \{W,B\}} p_{W,j} \geq \sum\limits_{j \in C\setminus \{W,B\}} p_{B,j}.$ 
    \item \label{cond:2} $\sum\limits_{j \in C\setminus \{W,B\}} p_{W,j} \geq \frac{(m-2)}{2}.$
\end{enumerate}

We use $y^{(j)}_i$ to denote $u_{i,j}$ for all $j \in C\setminus \{W,B\}.$ We use $b_i$ and $w_i$ for $u_{i,B}$ and $u_{i,W}$ respectively.

In terms of the empirical quantities, condition~\ref{cond:1} corresponds to the following:
\begin{equation} \label{eq:borda_0}
    \sum\limits_{j \in C\setminus \{W,B\}} \sum_{i \in V} \s(w_i - y^{(j)}_i) \geq \sum\limits_{j \in C\setminus \{W,B\}} \sum_{i \in V} \s(b_i - y^{(j)}_i) =  \sum\limits_{j \in C\setminus \{W,B\}} \sum_{i \in V} (1- \s(y^{(j)}_i - b_i)).
\end{equation}

Similar to the proof for \Cc, we will upper bound $ \sum\limits_{i \in V} \s(w_i - y^{(j)}_i)$ by using the upper bound on the slope of $\s$, and lower bound $\sum\limits_{i \in V} (1-\s(y^{(j)}_i - b_i)) $ by invoking Lemma~\ref{lem:linearization}.
\begin{equation}\label{eq:borda_1}
    \sum\limits_{i \in V} \s(w_i - y^{(j)}_i) \leq \sum\limits_{i \in V} \frac{\beta}{4} w_i + \s(- y^{(j)}_i) = \frac{\beta \U_W}{4} + \sum\limits_{i \in V} \s(- y^{(j)}_i).
\end{equation}
\begin{equation} \label{eq:borda_2}
    \sum\limits_{i \in V} (1-\s(y^{(j)}_i - b_i)) 
    \geq 
    \sum\limits_{i \in V} \s(-y^{(j)}_i) + \frac{(\U_B - \U_j)}{2}\frac{1-e^{-\beta}}{1+e^{-\beta}}.
\end{equation}

From Equations~\eqref{eq:borda_0},~\eqref{eq:borda_1} and~\eqref{eq:borda_2}, we get
\begin{equation} \label{eq:borda_6}
   (m-2) \frac{\beta \U_W}{4}  \geq \sum\limits_{j \in C\setminus \{W,B\}} \frac{(\U_B - \U_j)}{2}\frac{1-e^{-\beta}}{1+e^{-\beta}}.
\end{equation}
In terms of empirical quantities, condition~\ref{cond:2} corresponds to the following:

 \(  \sum\limits_{j \in C\setminus \{W,B\}} \sum_{i \in V} \s(w_i - y^{(j)}_i) \geq \frac{(m-2)n}{2}.\) 
 
Using Equation~\eqref{eq:borda_1}, we get
\begin{align} 
  & (m-2) \frac{\beta \U_W}{4} +  \sum\limits_{j \in C\setminus \{W,B\}} \sum_{i \in V} \s(- y^{(j)}_i)  \geq \frac{(m-2)n}{2}\\
   \implies & (m-2) \frac{\beta \U_W}{4} +  \sum\limits_{j \in C\setminus \{W,B\}} \sum_{i \in V} 1- \s(y^{(j)}_i)  \geq \frac{(m-2)n}{2}\\
   \implies & (m-2) \frac{\beta \U_W}{4} \geq \sum\limits_{j \in C\setminus \{W,B\}} \sum_{i \in V}  \s(y^{(j)}_i) - \frac{(m-2)n}{2}\\
   \intertext{Since $\s(x)$ is concave increasing in $x\geq 0,$ we have $\s(x) \geq \s(0) + x (\s(1)-\s(0)) = \frac{1}{2} + \frac{x (1-e^{-\beta}) }{2 (1+e^{-\beta})}$}
   &  (m-2) \frac{\beta \U_W}{4} \geq \sum\limits_{j \in C\setminus \{W,B\}} \frac{\U_j}{2} \frac{(1-e^{-\beta}) }{(1+e^{-\beta})}\label{eq:borda_5}
\end{align}
%


If $\sum_{j \in C\setminus\{W,B\}} \U_j \ge \left(\frac{m-2}{2}\right)\U_B,$
then \eqref{eq:borda_5} yields
\(
\frac{\beta \U_W}{2}\cdot \frac{1+e^{-\beta}}{1-e^{-\beta}} \;\ge\; \U_B.\)

Otherwise, if $\sum_{j \in C\setminus\{W,B\}} \U_j \le \left(\frac{m-2}{2}\right)\U_B,$
then \eqref{eq:borda_6} implies

\(
(m-2)\frac{\beta \U_W}{2}\cdot \frac{1+e^{-\beta}}{1-e^{-\beta}}
\;\ge\;
\sum_{j \in C\setminus\{W,B\}} (\U_B-\U_j)
\;=\;
(m-2)\U_B - \sum_{j \in C\setminus\{W,B\}} \U_j
\;\ge\;
\left(\frac{m-2}{2}\right)\U_B,
\)
and hence again $\frac{\U_B}{\U_W}\le \beta \frac{1+e^{-\beta}}{1-e^{-\beta}}$.

If $\textsc{Util}_B=0$ then $\textsc{Util}_j=0$ for all $j$ (utilities are nonnegative), so the
distortion is $1$. Otherwise $\textsc{Util}_B>0$, and by the strong law,
$\U_B/n\to \textsc{Util}_B$ and $\U_W/n\to \textsc{Util}_W$ almost surely.


Dividing by $n$ and letting $n\to\infty$ yields $\textsc{Util}_W \ge c_\beta \textsc{Util}_B>0$.
Therefore by the SLLN,
\[
\frac{\textsc{Util}_B}{\textsc{Util}_W}
=\lim_{n\to\infty}\frac{\U_B/n}{\U_W/n}
=\lim_{n\to\infty}\frac{\U_B}{\U_W},
\]
and the bound established above on $\U_B/\U_W$ implies the same bound on
$\textsc{Util}_B/\textsc{Util}_W$.
\end{proof}

\section{A Lower Bound for Finite-Precision Tournament-Based Rules}
\label{sec:universal-lb}

The results of the previous section show that tournament-based rules like \C and \B achieve distortion $O(\beta)$ independent of $m$. A natural question is whether \emph{any} deterministic rule can match the $\frac{\beta}{2}\frac{1+e^{-\beta}}{1-e^{-\beta}}$ distortion achieved by the randomized \ML rule~\citep{golz2025distortion}. We show that the answer is no for a broad class of rules.

Many practical voting rules, including \Cc, depend on the rankings only through the pairwise majority relation (who beats whom) or through vote margins rounded to some finite precision. We formalize this as the class of \emph{finite-precision tournament-based rules} and prove that every deterministic rule in this class has distortion at least $(\frac{5}{8} - \epsilon)\beta$ for any $\epsilon >0$, establishing a gap between deterministic and randomized tournament-based rules.

\begin{definition}[Finite-Precision Tournament-Based Rule]
\label{def:finite-precision}
Given a profile $\pi^n$, recall that $s_{j,j'} \in [0,1]$ denotes the empirical fraction of voters ranking $j$ above $j'$. Fix a precision parameter $\rho > 0$ and define the coarsened edge weight
\[
W_\rho(j,j') \;:=\; \left\lfloor \frac{s_{j,j'} - \frac{1}{2}}{\rho} \right\rfloor \in \mathbb{Z}.
\]
A voting rule is \emph{$\rho$-finite-precision tournament-based} if its output depends on $\pi^n$ only through $W_\rho$.
\end{definition}


\begin{theorem}[ Lower bound for finite-precision tournament-based rules]
\label{thm:detTournamentGraph_lb}
Fix $\rho > 0$. For every $\varepsilon > 0$, there exists $\beta_0 = \beta_0(\rho, \varepsilon) > 0$ such that for all $\beta \ge \beta_0$, the following holds: for every deterministic $\rho$-finite-precision tournament-based voting rule $f$, there exists an election instance with $m = 3$ candidates and utilities in $[0,1]$ satisfying
\(
\mathcal{D}(f) \;\ge\; \left(\frac{5}{8} - \varepsilon\right)\beta.
\)
\end{theorem}

\begin{proof}[Proof sketch]
We construct an instance with three candidates $B, W, Y$ and three voter types such that the coarsened tournament $W_\rho$ forms a symmetric cycle: $W_\rho(W,Y) = W_\rho(Y,B) = W_\rho(B,W)$. 
Since $W_\rho$ forms a symmetric 3-cycle, we construct three instances by cyclically permuting the assignment of candidates to roles $(B, W, Y)$. All three instances yield the same $W_\rho$, so the deterministic rule $f$ selects the same candidate index in each. In one of the three role assignments, that candidate plays the role of $W$ (the low-welfare candidate). It therefore suffices to construct an instance with a symmetric coarsened tournament where $\textsc{Util}_B / \textsc{Util}_W \geq (\frac{5}{8} - \varepsilon)\beta$.

We define three voter types: Type I (fraction $p$) has utilities $(u_B, u_Y, u_W) = (1-\eta, 1, 0)$; Type II (fraction $q$) has $(1, 0, \delta/\beta)$; Type III (fraction $1-p-q$) has $(0, 0, \delta/\beta)$, where $\eta, \delta \in (0,1)$ are parameters. The population pairwise probabilities are $p_{W,Y} = p \cdot \sigma_\beta(-1) + (1-p) \cdot \sigma_1(\delta)$, $p_{Y,B} = p \cdot \sigma_\beta(\eta) + q \cdot \sigma_\beta(-1) + \frac{1-p-q}{2}$, and $p_{B,W} = p \cdot \sigma_\beta(1-\eta) + q \cdot \sigma_\beta(1-\delta/\beta) + (1-p-q)(1-\sigma_1(\delta))$. We choose $p, q, \delta$ so that all three margins equal $\frac{1}{2} + \gamma$ in the limit $\beta \to \infty$, for some $\gamma \in (0, 2\rho/3)$. Solving this system yields $p = \frac{6\gamma + 4\gamma^2}{1+6\gamma}$, $q = \frac{4\gamma - 8\gamma^2}{1+6\gamma}$, and $\delta = 16\gamma + O(\gamma^2)$.

The population utilities are $\textsc{Util}_B = p(1-\eta) + q$ and $\textsc{Util}_W = (1-p)\delta/\beta$. For finite $\beta$, the true pairwise probabilities deviate from the idealized values by at most $E(\beta, \eta, \delta) = e^{-\beta} + e^{-\beta\eta} + e^{-\beta(1-\eta)} + e^{-\beta+\delta}$, which vanishes as $\beta \to \infty$. 
Choosing $\beta$ large enough that $E < \gamma/4$, all population margins 
lie within $\gamma/4$ of $\frac{1}{2} + \gamma$. By Hoeffding's inequality, 
the empirical margins concentrate within $\gamma/4$ of the population values 
with high probability for large $n$. Thus all empirical margins fall in 
$(\frac{1}{2} + \gamma/2, \frac{1}{2} + 3\gamma/2) \subset 
(\frac{1}{2}, \frac{1}{2} + \rho)$, where the first inclusion uses 
$\gamma > 0$ and the second uses $\gamma < 2\rho/3$.
 Hence $W_\rho(W,Y) = W_\rho(Y,B) = W_\rho(B,W) = 0$.
Analyzing the welfare ratio as $\gamma \to 0$: we have $p = 6\gamma + O(\gamma^2)$, $q = 4\gamma + O(\gamma^2)$, and $\delta = 16\gamma + O(\gamma^2)$. Substituting into $\frac{\textsc{Util}_B}{\textsc{Util}_W} = \beta \cdot \frac{p(1-\eta) + q}{\delta(1-p)}$ yields coefficient $\frac{(6\gamma)(1-\eta) + 4\gamma}{16\gamma} + O(\gamma) = \frac{10 - 6\eta}{16} + O(\gamma) = \frac{5-3\eta}{8} + O(\gamma)$. Setting $\eta = \frac{4\varepsilon}{3}$ gives $\frac{5-3\eta}{8} = \frac{5}{8} - \frac{\varepsilon}{2}$, and choosing $\gamma$ small enough yields distortion at least $(\frac{5}{8} - \varepsilon)\beta$. Full details appear in Appendix~\ref{app:universal-lb}.
\end{proof}

\section{Discussion}
We study utilitarian distortion under probabilistic voting where voters' behavior follows the Plackett-Luce model. Our results align the distortion framework with axiomatic intuitions. Particularly, tournament-based rules \C and \B have much better distortion guarantees than \Plu and \textsc{PluralityVeto}. \C satisfies the \textit{Condorcet criterion} (if a candidate beats all others in pairwise comparisons, then it must be the winner), while \Plu and \PluV do not. \B is used in practice for AI alignment \cite{siththaranjan2023distributional} and is shown to be the maximum likelihood estimator under some conditions~\cite{truchon2008borda}. It is the only rule that satisfies the axioms of neutrality, consistency, faithfulness, and the cancellation property simultaneously \cite{young1974axiomatization}. Recently \citet{maskin2025borda} showed that \B is the unique voting rule to satisfy a weakening of the \textit{independence of irrelevant alternatives} (IIA) axiom called \textit{modified independence of irrelevant alternatives (MIIA)}, in conjunction with other natural axioms.

There is a longstanding debate whether  \Cc, which is ``majoritarian'' (on pairwise comparisons) is ``better'' than \Bb, which is ``compromise-based'' (rewards high average performance across all comparisons). The series of papers from Saari and Risse, where Saari provides arguments in support of \B and Risse for \Cc, is particularly interesting~\cite{saari1985optimal,risse2001arrow, saari2006better, saari2003capturing, risse2005count,  risse2003democracy, saari2023selecting}. For metric distortion, \C is better both under deterministic \cite{anshelevich2015approximating} and probabilistic voting \cite{goyal2025metric}. This work presents an alternative perspective on this comparison, where neither dominates the other.
While we provide near-tight bounds for several voting rules, important open questions remain.

\begin{openq}
What is the optimal distortion for deterministic voting rules?
\end{openq}
We show a lower bound of $(\frac{5}{8}-\varepsilon)\beta$ for finite-precision tournament rules, but the optimal deterministic rule may be better. In metric distortion, the two-candidate case fortuitously yielded the universal lower bound, yet proving tightness required several years of effort~\citep{goel2017metric, munagala2019improved, gkatzelis2020resolving}. Under probabilistic voting, the problem is further complicated by the need to reason about distributions over rankings rather than fixed preference profiles. Even the $\frac{\beta}{2}\frac{1+e^{-\beta}}{1-e^{-\beta}}$ bound for randomized rules from prior work \citep{golz2025distortion} is not universal and holds only for rules that satisfy the PCLC (which excludes \Rr, for example). Identifying the optimal deterministic rule, or proving a stronger universal lower bound, remains open.

The next couple of research directions are about revisiting the models in this line of work.

\begin{openq}
How can we refine our frameworks for evaluating voting rules?
\end{openq}

An ideal framework would provide prescriptive advice to practitioners, for example, in political science or AI alignment. Considering probabilistic voting in the distortion model aligns the results with axiomatic intuitions in both the metric \cite{goyal2025metric} and utilitarian models.

However, there may be an opportunity to further refine these models. One limitation of the utilitarian model stems from the fact that while voters are assumed to compare alternatives on the \textit{difference} in utilities, the evaluation is on the \textit{ratio} of the social utilities. This creates an ``unstable'' evaluation criterion. 
For example, a clearly bad alternative (according to the ratio of utilities) may get a non-negligible number of votes. Consider the case with two alternatives, $W$ and $B$. If $p_{WB} = \frac{1}{2} - \delta$ for any constant $\delta >0,$ $\frac{\textsc{Util}_B}{\textsc{Util}_W}$ may be unbounded since $\textsc{Util}_W$ may be zero and $\textsc{Util}_B$ may be of the order $\frac{\delta n}{\beta}$. In this case, the difference $\textsc{Util}_B - \textsc{Util}_W$ is small and, arguably, more meaningful than the ratio. \citet{kahng2022worst} argues in favor of an additive distortion model instead of the one based on ratios. 
This example also demonstrates that the Weighted-Uncovered voting rule \citep{munagala2019improved} has unbounded distortion in this framework, which is counterintuitive, since it closely resembles \Cc. 

This problem is handled in the metric distortion model \citep{goyal2025metric},  by encoding a `scale-freeness' axiom and assuming that voters compare alternatives based on the \textit{ratio} of costs. However, the metric model has its own limitations.  \citet{stokes1963spatial} demonstrated over sixty years ago that spatial models fail to capture real electorate data. 
Other adjacent communities, including AI alignment via RLHF \cite{ouyang2022training, rafailov2023direct, casper2023open, siththaranjan2023distributional} and political science \cite{lijphart1994electoral, lijphart1999patterns, gallagher1991proportionality}, have distinct criteria for evaluating voting rules.
Combining insights from different research communities and creating a practically useful and theoretically grounded evaluation framework for voting rules is an interesting and timely problem.

\begin{openq}
Can the distortion framework be adapted to KL-constrained policy spaces?
\end{openq}
\citet{golz2025distortion} initiated the study of this question. A policy is a distribution over the candidates. The KL divergence bound restricts the set of policies to a ball $\mathcal{B}_{ref}$ around a reference policy $p_{ref}$. They demonstrate that \B has an exponential distortion in $\beta$ in this setting, whereas \ML has linear distortion. The proof constructs instances in which \B compares policies in $\mathcal{B}_{ref}$ primarily with candidates with small weights in $p_{ref} $, and infrequently with policies it has to compete with. Whereas \ML decides the winning policy by comparing policies within $\mathcal{B}_{ref}$ with each other. Since the evaluation criterion is the comparison with the best policy in $\mathcal{B}_{ref},$ \B has a very high distortion in this framework.
It is interesting to study cases where the distribution of comparison pairs does not exhibit a strong mismatch with $p_{ref}$. here \B will likely perform better.

\bibliographystyle{plainnat}
\bibliography{refs}

\appendix

\section{Proof of Theorem~\ref{thm:detTournamentGraph_lb}}
\label{app:universal-lb}

\begin{proof}
Fix a deterministic $\rho$-finite-precision tournament-based rule $f$. We construct three candidates $B, W, Y$ and a voter distribution such that the coarsened tournament forms a symmetric 3-cycle, forcing $f$ to select a low-welfare candidate.

\paragraph{Symmetry argument.}
Suppose the edge-weight matrix $W_\rho$ satisfies: $W_\rho(W,Y) = W_\rho(Y,B) = W_\rho(B,W) = \ell$ for some integer $\ell \ge 0$. We construct three election instances by cyclically permuting the assignment of the three candidates to the roles $(B, W, Y)$ in the utility construction below. Since all three instances produce the same symmetric $W_\rho$, the deterministic rule $f$ selects the same candidate index $k$ in each. In one of the three role assignments, candidate $k$ plays the role of $W$ (the low-welfare candidate). It therefore suffices to construct an instance with a symmetric coarsened tournament where $\textsc{Util}_B / \textsc{Util}_W \ge (\frac{5}{8} - \varepsilon)\beta$.

\paragraph{Three-type construction.}
Let $\gamma \in (0, 2\rho/3)$ and $\eta \in (0, 1)$ be parameters to be specified. Define three voter types with utilities for $(B, Y, W)$:
\begin{center}
\begin{tabular}{lccc}
& $u_B$ & $u_Y$ & $u_W$ \\[2pt]
\textbf{Type I} (fraction $p$): & $1 - \eta$ & $1$ & $0$ \\
\textbf{Type II} (fraction $q$): & $1$ & $0$ & $\delta/\beta$ \\
\textbf{Type III} (fraction $1-p-q$): & $0$ & $0$ & $\delta/\beta$
\end{tabular}
\end{center}
Here $\delta \in (0,1]$ is determined below. All utilities lie in $[0,1]$.

The logistic sigmoid function satisfies, for all $x > 0$ and $\beta > 0$:
\begin{equation}\label{eq:sigmoid-bounds}
1 - e^{-\beta x} \;\le\; \sigma_\beta(x) \;\le\; 1, \qquad 0 \;\le\; \sigma_\beta(-x) \;\le\; e^{-\beta x}.
\end{equation}
Define $s := \sigma_1(\delta) = (1 + e^{-\delta})^{-1} \in (\frac{1}{2}, 1)$. The population pairwise win-probabilities are:
\begin{align}
p_{W,Y} &= p \cdot \sigma_\beta(-1) + (1-p) \cdot s, \label{eq:pWY} \\
p_{Y,B} &= p \cdot \sigma_\beta(\eta) + q \cdot \sigma_\beta(-1) + \frac{1-p-q}{2}, \label{eq:pYB} \\
p_{B,W} &= p \cdot \sigma_\beta(1-\eta) + q \cdot \sigma_\beta(1 - \delta/\beta) + (1-p-q)(1-s). \label{eq:pBW}
\end{align}

\paragraph{Idealized system.}
We determine $p, q, \delta$ by considering the limiting system as $\beta \to \infty$, where $\sigma_\beta(-1) \to 0$, $\sigma_\beta(\eta) \to 1$, $\sigma_\beta(1-\eta) \to 1$, and $\sigma_\beta(1 - \delta/\beta) \to 1$. Requiring all three margins to equal $\frac{1}{2} + \gamma$:
\begin{align}
(1-p)s &= \frac{1}{2} + \gamma, \label{eq:ideal1} \\
p + \frac{1-p-q}{2} &= \frac{1}{2} + \gamma, \label{eq:ideal2} \\
p + q + (1-p-q)(1-s) &= \frac{1}{2} + \gamma. \label{eq:ideal3}
\end{align}
From \eqref{eq:ideal2}, we obtain $q = p - 2\gamma$. Substituting into \eqref{eq:ideal3} and combining with \eqref{eq:ideal1} yields:
\begin{equation}\label{eq:params}
p = \frac{6\gamma + 4\gamma^2}{1 + 6\gamma}, \qquad q = \frac{4\gamma - 8\gamma^2}{1 + 6\gamma}, \qquad s = \frac{\frac{1}{2} + 4\gamma + 6\gamma^2}{1 - 4\gamma^2}.
\end{equation}
We set $\delta = \ln\bigl(\frac{s}{1-s}\bigr)$. 

The true probabilities \eqref{eq:pWY}--\eqref{eq:pBW} deviate from the idealized values \eqref{eq:ideal1}--\eqref{eq:ideal3}. Using \eqref{eq:sigmoid-bounds}:
\begin{align*}
\bigl| p_{W,Y} - (1-p)s \bigr| &\le p \cdot e^{-\beta}, \\
\bigl| p_{Y,B} - \bigl(p + \tfrac{1-p-q}{2}\bigr) \bigr| &\le p \cdot e^{-\beta\eta} + q \cdot e^{-\beta}, \\
\bigl| p_{B,W} - \bigl(p + q + (1-p-q)(1-s)\bigr) \bigr| &\le p \cdot e^{-\beta(1-\eta)} + q \cdot e^{-\beta(1 - \delta/\beta)}.
\end{align*}
Define:
\[
E(\beta, \eta, \delta) \;:=\; e^{-\beta} + e^{-\beta\eta} + e^{-\beta(1-\eta)} + e^{-\beta + \delta}.
\]
Then $\bigl|p_{j,j'} - (\frac{1}{2} + \gamma)\bigr| \le E(\beta, \eta, \delta)$ for each pair $(j, j')$. For fixed $\eta \in (0,1)$ and $\delta \in (0,1)$, the function $E(\beta, \eta, \delta) \to 0$ as $\beta \to \infty$.

\paragraph{Coarsening via concentration.}
For the coarsened edge weights to be identical, it suffices that all three 
empirical margins $s_{j,j'}$ lie in the interval $[\frac{1}{2}, \frac{1}{2} + \rho)$.

First, choose $\beta_1 = \beta_1(\gamma, \eta, \delta)$ such that 
$E(\beta, \eta, \delta) < \gamma/4$ for all $\beta \ge \beta_1$. Then all 
population margins lie in 
$(\frac{1}{2} + \frac{3\gamma}{4}, \frac{1}{2} + \frac{5\gamma}{4})$.

Second, for $n$ voters drawn i.i.d., Hoeffding's inequality gives 
$\Pr(|s_{j,j'} - p_{j,j'}| > t) \le 2\exp(-2nt^2)$. Taking $t = \gamma/4$ 
and applying a union bound over the three pairs, for 
$n \ge 32\gamma^{-2}\ln(6/\alpha)$, all empirical margins lie within 
$\gamma/4$ of their population values with probability at least $1 - \alpha$. 
Combined with the population bound, all empirical margins fall in 
$(\frac{1}{2} + \frac{\gamma}{2}, \frac{1}{2} + \frac{3\gamma}{2})$. 
Since $\gamma > 0$, the lower endpoint exceeds $\frac{1}{2}$; 
since $\gamma < 2\rho/3$, the upper endpoint $\frac{1}{2} + \frac{3\gamma}{2} < \frac{1}{2} + \rho$. 
Hence $W_\rho(W,Y) = W_\rho(Y,B) = W_\rho(B,W) = 0$, and the coarsened 
tournament is symmetric.

\paragraph{Welfare ratio.}
The population utilities are:
\[
\textsc{Util}_B = p(1-\eta) + q, \qquad \textsc{Util}_W = (1-p) \cdot \frac{\delta}{\beta}.
\]
When $f$ outputs $W$, the distortion satisfies:
\begin{equation}\label{eq:distortion}
\frac{\textsc{Util}_B}{\textsc{Util}_W} = \beta \cdot \frac{p(1-\eta) + q}{\delta(1-p)}.
\end{equation}

\paragraph{Asymptotic analysis.}
We analyze the coefficient $C := \frac{p(1-\eta) + q}{\delta(1-p)}$ for small $\gamma$. From \eqref{eq:params}:
\[
p = 6\gamma + O(\gamma^2), \quad q = 4\gamma + O(\gamma^2), \quad 1 - p = 1 - 6\gamma + O(\gamma^2), \quad s = \frac{1}{2} + 4\gamma + O(\gamma^2).
\]
The Taylor expansion $\sigma_1(\delta) = \frac{1}{2} + \frac{\delta}{4} + O(\delta^3)$ combined with $s = \frac{1}{2} + 4\gamma + O(\gamma^2)$ yields $\delta = 16\gamma + O(\gamma^2)$. Substituting:
\[
C = \frac{(6\gamma)(1-\eta) + 4\gamma + O(\gamma^2)}{(16\gamma)(1 - 6\gamma + O(\gamma^2)) + O(\gamma^2)} = \frac{(10 - 6\eta)\gamma + O(\gamma^2)}{16\gamma + O(\gamma^2)} = \frac{5 - 3\eta}{8} + O(\gamma).
\]

\paragraph{Parameter selection.}
Given $\varepsilon > 0$, choose parameters as follows. Set 
$\eta = \frac{4\varepsilon}{3}$, so that $\frac{5 - 3\eta}{8} = \frac{5}{8} - \frac{\varepsilon}{2}$. 
Choose $\gamma_0 > 0$ small enough that: (i) the $O(\gamma)$ remainder 
satisfies $|O(\gamma)| \le \frac{\varepsilon}{4}$ for $\gamma \le \gamma_0$, 
(ii) all feasibility conditions on $p, q, s, \delta$ hold, and 
(iii) $\gamma_0 < 2\rho/3$. Set $\gamma = \gamma_0$, which determines 
$\delta$ via \eqref{eq:params}. Finally, choose $\beta_0$ large enough that 
$E(\beta_0, \eta, \delta) < \gamma/4$.
For all $\beta \ge \beta_0$:
\[
\frac{\textsc{Util}_B}{\textsc{Util}_W} \;\ge\; \beta \left(\frac{5}{8} - \frac{\varepsilon}{2} - \frac{\varepsilon}{4}\right) \;>\; \left(\frac{5}{8} - \varepsilon\right)\beta.
\]
Combined with the symmetry argument, this completes the proof.
\end{proof}

\section{Proofs of Theorem~\ref{thm:pluV_lb} and Corollary~\ref{corr:ppv}}
\label{app:pluv}

\textbf{Proof of \PluV Lower Bound Theorem~\ref{thm:pluV_lb}}
\begin{proof}
We construct a distribution $\mathcal{V}$ over two voter types, each with probability $\frac{1}{2}$.

\emph{Type 1 voters (probability $\frac{1}{2}$):} $u_{i,B} = 0$ and $u_{i,j} = \frac{2}{\beta \ln m}$ for all $j \neq B$. Let $V'$ be the set of Type 1 voters.

\emph{Type 2 voters (probability $\frac{1}{2}$):} $u_{i,B} = \frac{1}{2}$. Additionally, type 2 voters are partitioned into $\binom{m-1}{\lfloor \frac{m-1}{\beta \ln m} \rfloor}$ sets, each with probability $\frac{1}{2}\times\frac{1}{\binom{m-1}{\lfloor \frac{m-1}{\beta \ln m} \rfloor}}$. Each set corresponds to a subset $S_i \subset C \setminus \{B\}$ of size $\lfloor \frac{m-1}{\beta \ln m} \rfloor$ such that $u_{i,j} = 1$ for $j \in S_i$ and $u_{i,j} = 0$ for $j \notin S_i \cup \{B\}$. For each candidate $j \neq B,$ $u_{i,j}=1$ for at most a $\frac{1}{\beta \ln m}$ fraction of Type 2 voters. Let $\bar{V}'$ be the set of Type 2 voters.

  Observe that $\textsc{Util}_{B} = \frac{n}{4}.$ By symmetry, for all other candidates, the social utility is at most $\frac{3n}{2\beta \ln m}.$ 
  
  If $B$ loses, the distortion is at least $\frac{\beta \ln m}{6}.$

  We now show that $B$ loses the election under \PluV with probability 1 as $n\rightarrow \infty$.

Recall that the probability that $B$ is at the top of voter $i$'s ranking is: 
$$\probability(\top_i = B) = \frac{e^{\beta u_{i,B}}}{\sum\limits_{j \in C} e^{\beta u_{i,j}}}.$$ 
We set $ \beta \leq \frac{m}{3\ln m}.$ Further, we set $\beta  \geq 2\ln \beta + 2 \ln \ln m.$ Summing over all voters for $B$'s top votes:
\begin{align}
    \sum_{i \in V} \probability(\top_i = B) &= \sum_{i \in V'} \probability(\top_i = B) + \sum_{i \in \bar{V}'} \probability(\top_i = B) \\
    &= \frac{n}{2} \frac{1}{1 + (m-1) e^{\frac{2}{\ln m}}} 
    + \frac{n}{2} \frac{e^{\frac{\beta}{2}}} {e^{\frac{\beta}{2}} 
    +\lfloor \frac{m-1}{\beta \ln m}\rfloor e^{\beta} + m-1 -\lfloor \frac{m-1}{\beta \ln m}\rfloor } \\
    \intertext{Since $e^{\frac{2}{\ln m}} > 1$ we have $1 + (m-1) e^{\frac{2}{\ln m}} \geq m.$ Also, since $\beta \leq \frac{m}{3\ln m},$ and $m\geq10,$ we have $\lfloor \frac{m-1}{\beta \ln m}\rfloor \geq \frac{2m}{3\beta \ln m}.$ Therefore,   }
    &\leq \frac{n}{2m}   + \frac{3n  e^{-\frac{\beta}{2}} \beta \ln m }{4m}. \\
    \intertext{Since $\beta  \geq 2\ln \beta + 2 \ln \ln m,$ we have $e^{-\frac{\beta}{2}} \beta \ln m \leq 1$. This implies,}
    \sum_{i \in V} \probability(\top_i = B) &\leq \frac{5n}{4m}. \label{eq:top_ub}
\end{align}

Whereas, for the bottom votes for $B$, which lead to vetoes, we only need to account for the votes from voters in $V'.$ Let $\bot_i$ denote the bottom-ranked candidate in voter $i$'s ranking.

\begin{align}
    \sum_{i \in V} \probability(\bot_i = B) &\geq  \sum_{i \in V'} \probability(\bot_i = B)  
    = \frac{n}{2} \prod_{k = 1}^{m-1} \frac{(m-k) e^{\frac{2}{\ln m}}} {(m-k) e^{\frac{2}{\ln m}} + 1} \\
    &=  \frac{n}{2} \prod_{k = 1}^{m-1} \frac{1} {1 + \frac{e^{-\frac{2}{\ln m}}}{(m-k)}} \\
    &\geq \frac{n}{2} \exp\left(- \sum_{k=1}^{m-1} \ln \left(1+  \frac{e^{-\frac{2}{\ln m}}}{k} \right) \right)\\
    &\geq \frac{n}{2} \exp\left(- \sum_{k=1}^{m-1}  \left( \frac{e^{-\frac{2}{\ln m}}}{k} \right) \right)\\
    &\geq \frac{n}{2} \exp\left(- e^{-\frac{2}{\ln m}} (\ln m +1) \right)\\
    &= \frac{n}{2} \exp\left( \ln \left((em)^{- e^{-\frac{2}{\ln m}}} \right)\right)\\
    &=  \frac{n}{2}  \frac{1}{(em)^{e^{-\frac{2}{\ln m}}} }
\end{align}

For $B$ to lose the election, a sufficient condition is that it gets more bottom votes than top votes. This is because the all of $B's$ tokens are removed before the final round. Note that in this scenario, the sequence of elimination is irrelevant because when the algorithm reaches a voter with bottom vote of $B,$ it removes one token from $B$ regardless of the previous operations.
Therefore a sufficient condition is:
\begin{align}
     \frac{n}{2}  \frac{1}{(em)^{e^{-\frac{2}{\ln m}}} }  &\geq \frac{5n}{4m}. \\
     \implies    \frac{1}{(em)^{e^{-\frac{2}{\ln m}}} }  &\geq \frac{5}{2m}. \\
     \implies    m^{1 - e^{-\frac{2}{\ln m}}} &\geq e^{ e^{-\frac{2}{\ln m}}} \left(\frac{5}{2} \right). \\
     \implies    (\ln m) (1 - e^{-\frac{2}{\ln m}}) &\geq  e^{-\frac{2}{\ln m}} + \ln  \left(\frac{5}{2} \right).
\end{align}
The last inequality holds for $m\geq 10$. Therefore, $B$ loses the election. This completes the proof.
\end{proof}

This same construction also gives us an identical lower bound for \PPVv. \\

\textbf{Proof of \PPV Lower Bound Corollary~\ref{corr:ppv}}
\begin{proof}

Recall \PPV is defined as follows. 
For any positive constant $\alpha,$
consider only those candidates who have at least $\frac{\alpha n}{(6+\alpha)m}$ top votes. Call this set $C_{\textsc{pruned}}.$ Project the rankings $\pi^n$ to $C_{\textsc{pruned}}$. Call it $\pi^n_{\textsc{pruned}}.$
Run \PluV on $\pi^n_{\textsc{pruned}}.$

   For all $m\geq 10,$ we have $\sum_{i \in V} \probability(\top_i = B) \geq \sum_{i \in V'} \probability(\top_i = B)  =  \frac{n}{2} \frac{1}{1 + (m-1) e^{\frac{2}{\ln m}}} \geq \frac{n}{2em}.$  
   By Equation~\eqref{eq:top_ub}, all other candidates, by symmetry, also get at least $ \frac{n - \frac{5n}{4m}}{m-1} \geq \frac{n}{2m}$ top votes in expectation. 
   Therefore, for any $\alpha \leq 1.3,$ no candidates are removed in pruning, and $C_{\textsc{pruned}} = C.$ The same distortion as $\PluV$ is attained.
\end{proof}

\end{document}